\documentclass[notitlepage, nofootinbib, no reprint, amsmath, amssymb, aps, twocolumn, pra]{revtex4-2}

\usepackage{amssymb}
\usepackage{hyperref}
\usepackage{bm}
\usepackage{soul}
\usepackage{dsfont}
\usepackage{amsthm}
\usepackage{amsmath}
\usepackage{verbatim} 
\usepackage{qcircuit}
\usepackage{tabularx}
\usepackage{color,soul}
\usepackage{esint}
\usepackage{physics}
\usepackage{braket}
\usepackage{bbold}
\usepackage{amssymb}
\usepackage{float}
\usepackage[normalem]{ulem}

\usepackage{todonotes}

\def\R{\mathbb{R}}

\def\Id{\mathbb{1}}

\def\ra{\rho_\alpha}
\def\wls{w_\mathrm{ls}}
\def\wvar{w_\mathrm{var}}

\newcommand\psia{\ket{\psi_\alpha}}

\usepackage{mathtools}

\DeclarePairedDelimiterX{\infdivx}[2]{(}{)}{%
  #1\;\delimsize\|\;#2%
}
\newcommand{\infdiv}{D_f\infdivx}

\newtheorem{theorem}{Theorem}

\newtheorem{lemma}[theorem]{Lemma}

\theoremstyle{remark}
\newtheorem*{remark}{Remark}

\newtheorem{observation}{Observation}

\begin{document}

\author{Andrey Kardashin}
\email{kardashin.andrey@gmail.com} 
\affiliation{Skolkovo Institute of Science and Technology, Moscow, Russia}
\altaffiliation{Former affiliation. Current address: Donostia International Physics Center, San Sebastián/Donostia, Spain.}

\author{Konstantin Antipin}
\affiliation{Faculty of Physics, M.V. Lomonosov Moscow State University,\\
Leninskie gory, GSP-1, Moscow 119991,  Russia}
\affiliation{Skolkovo Institute of Science and Technology, Moscow, Russia }

\title{On measurement-dependent variance in quantum neural networks}

\begin{abstract}

    Variational quantum circuits have become a widely used tool for performing quantum machine learning (QML) tasks on labeled quantum states. 
    In some specific tasks or for specific variational ans\"atze, one may perform measurements on a restricted part of the overall input state. 
    This is the case for, e.g., quantum convolutional neural networks (QCNNs), where after each layer of the circuit a subset of qubits of the processed state is measured or traced out, and at the end of the network one typically measures a local observable. 
    In this work, we demonstrate that measuring observables with restricted support results in larger label prediction variance in regression QML tasks. 
    We show that the reason for this is, essentially, the number of distinct eigenvalues of the observable one measures after the application of a variational circuit.

\end{abstract}

\maketitle

\section{Introduction}

    Quantum machine learning (QML)~\cite{schuld2015introduction, biamonte2017quantum, lloyd2018quantum,schuld2019quantum,mitarai2018quantum} has increasingly focused on learning directly from quantum data, i.e., scenarios in which inputs are quantum states rather than classical feature vectors~\cite{carrasquilla2017machine,van2017learning, uvarov2020machine,patterson2021quantum,gibbs2024dynamical,qiu2019detecting,sanavio2023entanglement,kardashin2025predicting}. 
    Variational quantum circuits~\cite{benedetti2019parameterized,schuld2020circuit} provide a flexible framework for such tasks. 
    That is, in regression problems, one can consider a parameterized unitary transformation $U_{\boldsymbol{\theta}}$ of an observable $M$, $M_{\boldsymbol{\theta}} \equiv U^\dagger_{\boldsymbol{\theta}}\, M\, U_{\boldsymbol{\theta}}$, whose expectation value on a labeled input state $\rho_\alpha$ constitutes an estimation of the label $\alpha$.
    This estimate is obtained by repeated measurements, and its sampling variance~(the square of the standard deviation of the mean) is dictated by the variance of $M_{\boldsymbol{\theta}}$.
    This variance is therefore very important in relation to achieving higher precision of the result given fixed number of measurement shots.
    
    As we show in this work, the structure of $M_{\boldsymbol{\theta}}$ and hence its variance is significantly controlled by the support size of the initial observable $M$. 
    As a result, observables acting on many qubits can enable reduced variance for fixed shot budgets. 
    In contrast, the use of observables with more restricted support ---such as single-qubit Pauli operators--- might lead to higher estimation variance.
    
    This observation has important implications for such architectures as quantum convolutional neural networks (QCNNs), in which one commonly measures a few qubits.
    These and similar networks have demonstrated strong performance in classification and phase-recognition tasks \cite{grant2018hierarchical, cong2019quantum, huggins2019towards, pesah2021absence}, as well as solving various regression problems \cite{wu2021scrambling, shen2020information, nagano2023quantum, umeano2023can}.
    QCNNs employ a hierarchically structured circuit that culminates in a local measurement, often on a single qubit. 
    While this design yields shallow circuits and favorable scaling, it also restricts the structural richness of the readout observable, potentially increasing sampling variance. 
    Our analysis makes this trade-off explicit and quantifies how measurement constraints affect prediction variance in quantum-data regression. 
    Interestingly enough, QCNNs have recently gained considerable attention also in relation to their classical simulability, which is believed to arise from their ability to extract information encoded only in low-weight observables of their input states~\cite{zoeqcnnclass2024}.

    Recent advances in measurement optimization for variational algorithms ---including measurement grouping~\cite{zhu2024optimizing}, classical-shadow-based strategies~\cite{jnane2024quantum, huang2020predicting}, and machine-learning methods~\cite{torlai2020precise}--- offer potential variance reduction techniques. 
    Understanding the relationship between the readout  observable properties and estimator variance is thus essential to designing QML architectures that remain both sample-efficient and experimentally feasible.

    In the present work, we analyze the dependence of the variance of an observable on its structural properties, such as the number of qubits it is supported on, and the degeneracy of its spectrum. 
    This phenomenon is illustrated on a number of regression QML tasks, which include finding the weight in convex combination of states, and also predicting the parameters  of several paradigmatic local Hamiltonian models. 
    We stress that our considerations are not restricted to QCNN architectures alone, and in the next sections we show that similar variance effects manifest themselves in QML with other variational ans\"atze.

    The work is structured as follows.
    In Section~\ref{sec:problem_statement}, we state the regression problem and introduce the notation used in the paper.
    Section~\ref{sec:methods} describes the methods employed in the work, including variational quantum computing framework, as well as the classical and quantum Fisher information applied for assessing the prediction variance.
    Main results of our work are shown in Section~\ref{sec:results}, where we give analytical expressions for the variance for two regression tasks, and support our claims with numerical experiments.
    Section~\ref{sec:conclusion} concludes the paper.
    In Appendices, one can find detailed derivations of our analytical results, descriptions of the variational ans\"atze applied, and additional numerical results.

\section{Problem statement}
\label{sec:problem_statement}

    Consider a set of the form $\mathcal{T} = \left\{ (\rho_{\alpha_j}, \alpha_j ) \right\}_{j=1}^{T}$, where $\rho_{\alpha_j}$ are quantum states labeled by $\alpha_j \in \R$.
    Our goal is to use $\mathcal{T}$ for learning to predict the label $\alpha$ for a given datum $\rho_\alpha$ not present in $\mathcal{T}$.
    Essentially, we want to solve a regression problem, but with a peculiarity that $\rho_\alpha$ are quantum states. 

    Since the labeled data points $\rho_\alpha$ are quantum states, it would be natural to seek for a prediction $\mathsf{a}$ of the label $\alpha$ as the expectation of an observable $M$ measured in $\rho_\alpha$.
    That is, $\mathsf{a} = \langle M \rangle_{\rho_\alpha} \equiv \Tr M \rho_\alpha = \alpha + b_\alpha$ with $b_\alpha$ being a bias.
    Since $M$ is a Hermitian operator, it can be represented as 
    \begin{equation}
        \label{eq:obs-no_par}
        M = \sum_{i=1}^D \lambda_i \Lambda_i,
    \end{equation}
    where $\boldsymbol{\Lambda}=\{\Lambda_i\}_{i=1}^D$ are orthogonal projectors, and $\boldsymbol{\lambda}=\{\lambda_i\}_{i=1}^D$ are real coefficients.

    Our goal is finding an observable $M$ which gives accurate predictions $\mathsf{a}=\langle M \rangle_{\rho_\alpha} = \alpha + b_\alpha$ for $\alpha$ with small bias $b_\alpha$ and presumably low variance $\Delta^2_{\rho_\alpha} M \equiv \langle M^2 \rangle_{\rho_\alpha} - \langle M \rangle_{\rho_\alpha}^2$.
    As stated earlier, this variance defines the number of measurement shots one has to conduct for achieving a given estimation precision  \cite{kreplin2024reduction}.
    In our work, we show that depending on the connection between $\rho_\alpha$ and $\alpha$, for reducing the variance one may need a different number of terms $D$ in \eqref{eq:obs-no_par}, or, similarly, different dimensiounalities of the projectors $\Lambda_i$.

\section{Methods}
\label{sec:methods}

    Following the methodology presented in \cite{kardashin2025predicting}, we now show how one can find an observable $M$ using the variational quantum computing approach \cite{cerezo2021variational}.
    That is, one can parametrize the projectors in \eqref{eq:obs-no_par} as $\Lambda_i(\boldsymbol{\theta}) = U^\dagger_{\boldsymbol{\theta}} \big(\Id^{\otimes(n-m)} \otimes \ketbra{i}{i} \big) U_{\boldsymbol{\theta}}$, where $U_{\boldsymbol{\theta}}$ is a variational quantum circuit, i.e., a unitary operator parametrized by $\boldsymbol{\theta} \subset \R$, $\Id$ is the single-qubit identity operator, and $\ketbra{i}{i}$ is the projector onto the $i$th state of the computational basis of $m \leqslant n$ qubits.
    Therefore, the parametrized observable takes the form
    \begin{equation}
        \label{eq:obs-par}
        M_{\boldsymbol{\lambda}, \boldsymbol{\theta}} = \sum_{i=1}^{2^m} \lambda_i \; U^\dagger_{\boldsymbol{\theta}} \left(\Id^{\otimes(n-m)} \otimes \ketbra{i}{i} \right) U_{\boldsymbol{\theta}}.
    \end{equation}
    
    Schematically, the process of measuring $M_{\boldsymbol{\lambda}, \boldsymbol{\theta}}$ in an $n$-qubit labeled state $\rho_\alpha$ is depicted in Fig.~\ref{fig:circuit}.
    One can view this as measuring the last $m \leqslant n$ qubits of the transformed labeled state $\rho_\alpha(\boldsymbol{\theta}) \equiv U_{\boldsymbol{\theta}}\rho_\alpha U^\dagger_{\boldsymbol{\theta}}$ in the computational basis, with probability $p_i (\boldsymbol{\theta})= \Tr\! \left[\left( \Id^{\otimes (n - m)} \otimes \ketbra{i} \right) \rho_\alpha(\boldsymbol{\theta})\right]$ obtaining the outcome $i$ associated with $\lambda_i$, and evaluating the expectation as $\langle  M_{\boldsymbol{\lambda}, \boldsymbol{\theta}} \rangle_{\rho_\alpha} = \sum_{i=1}^{2^m} p_i(\boldsymbol{\theta}) \lambda_i$.
    Note that we can also compute the variance as $\Delta_{\rho_\alpha}^2 M_{\boldsymbol{\lambda}, \boldsymbol{\theta}} = \sum_{i=1}^{2^m} p_i(\boldsymbol{\theta}) \lambda_i^2 - \big(\sum_i p_i(\boldsymbol{\theta}) \lambda_i\big)^2$. 

    As an alternative to measuring $m \leqslant n$ qubits of the labeled state $\rho_\alpha$, the qubits to be measured can be introduced as the auxiliary ones, as also shown in Fig.~\ref{fig:circuit}.
    That is, attaching the $m_{\mathrm{a}}$ qubits in the state $\ketbra{0}$, one can find a variational circuit $U_{\boldsymbol{\theta}}$ acting on the joint state $\rho_\alpha\otimes\ketbra{0}^{\otimes m_{\mathrm{a}}}$ such that it reproduces the measurement results of the observable \eqref{eq:obs-no_par} with arbitrary eigenprojectors $\Lambda_i$, namely,
    \begin{equation}
        \label{eq:naimark-variational}
        \Tr \Lambda_i \rho_\alpha = \Tr\Big[U^\dagger_{\boldsymbol{\theta}} \big(\Id \otimes \ketbra{i}\big) U_{\boldsymbol{\theta}} \, \big(\rho_\alpha\otimes\ketbra{0}^{\otimes m_{\mathrm{a}}}\big) \Big].
    \end{equation}
    This technique is known as the Naimark's extension \cite{rethinasamy2023estimating}.

    Given a training set $\mathcal{T} = \left\{ (\rho_{\alpha_j}, \alpha_j ) \right\}_{j=1}^{T}$, the optimal parameters $(\boldsymbol{\lambda}^*, \boldsymbol{\theta}^*)$ for $M_{\boldsymbol{\lambda}, \boldsymbol{\theta}}$ can be found by solving
    \begin{multline}
        \label{eq:ls_min}
        (\boldsymbol{\lambda}^*, \boldsymbol{\theta}^*) = \arg\min_{\boldsymbol{\lambda}, \boldsymbol{\theta}} 
        \Bigg[ 
        w_\mathrm{ls} \sum_{j=1}^T \Big(\alpha_j - \langle M_{\boldsymbol{\lambda}, \boldsymbol{\theta}} \rangle_{\rho_{\alpha_j}}  \Big)^2 
        \\+
        w_\mathrm{var} \sum_{j=1}^T \Delta_{\rho_{\alpha_j}}^2 M_{\boldsymbol{\lambda}, \boldsymbol{\theta}}
        \Bigg], 
    \end{multline}
    where $\wls, \wvar > 0$.
    That is, we simultaneously minimize the least squares between the labels $\alpha_j$ and our predictions $\mathsf{a}_j = \langle M_{\boldsymbol{\lambda}, \boldsymbol{\theta}} \rangle_{\rho_{\alpha_j}}$ with weight $\wls$, and the sum of variances $\Delta^2_{\rho_{\alpha_j}} M_{\boldsymbol{\lambda}, \boldsymbol{\theta}}$ with weight $\wvar$.

    \begin{figure}
        \centering
        \begin{equation*}
            \Qcircuit @C=1em @R=1em{
                & \qw & \qw/_{n-m} & \qw & \multigate{1}{U_{\boldsymbol{\theta}}} &        &        & \\
                & \qw & \qw/_m     & \qw & \ghost{U_{\boldsymbol{\theta}}}        & \qw/_m & \meter & \;\;\qquad \overset{\,p_i}{\underset{}{\longrightarrow}} i \mapsto \lambda_i 
                \inputgroupv{1}{2}{.75em}{1.em}{\rho_\alpha}
            }
        \end{equation*}
        \vspace{1em}
        \begin{equation*}
            \Qcircuit @C=1em @R=1em{
                \rho_\alpha   & & & \qw/_n & \multigate{1}{U_{\boldsymbol{\theta}}} &     &        & \\
                \ketbra{0}{0}^{\otimes m_{\mathrm{a}}} & & & \qw/_{m_{\mathrm{a}}} &        \ghost{U_{\boldsymbol{\theta}}} & \qw/_{m_{\mathrm{a}}} & \meter & \;\;\qquad \overset{\,p_i}{\underset{}{\longrightarrow}} i \mapsto \lambda_i
            }
        \end{equation*}
        \caption{
            Upper: Schematic representation of measuring the observable \eqref{eq:obs-par} in an $n$-qubit state $\rho_\alpha$, with $m$ qubits being measured.
            Lower: Instead of measuring the $m$ qubits of $\rho_\alpha$, one can introduce $m_{\mathrm{a}}$ auxiliary qubits via the Naimark's extension as in \eqref{eq:naimark-variational}, which allows obtaining the eigenprojectors of arbitrary ranks.
        }
        \label{fig:circuit}
    \end{figure}
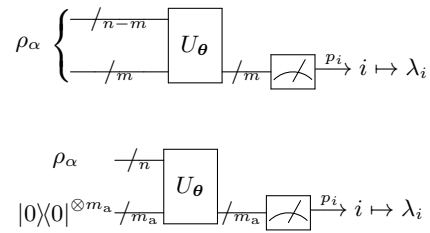

    Looking at \eqref{eq:obs-par}, one may notice that the (eigen)projectors of the observable $M_{\boldsymbol{\lambda}, \boldsymbol{\theta}}$ have the rank $2^{n-m}$.
    This rank can be tuned by measuring different numbers of qubits $m$.
    Particularly, one can employ $U_{\boldsymbol{\theta}}$ to be a QCNN, for which typically $m=1$, and which produces two projectors of dimensionality $2^{n-1}$.
    As mentioned earlier, and as we show in our work, this may affect the quality of the observable found for predicting the label $\alpha$ of $\rho_\alpha$.
    Namely, smaller $m$ may still produce observables \eqref{eq:obs-par} giving label predictions $\mathsf{a} = \langle M \rangle_{\rho_{\alpha}} = \alpha + b_\alpha$ with smaller bias $b_\alpha$, but with larger variance $\Delta_{\rho_{\alpha}}^2 M$.

    For this variance, one can write \cite{shettell2022quantum, kardashin2025predicting}
    \begin{equation}
        \label{eq:var-qcrb}
        \frac{\Delta_{\rho_\alpha}^2 M} {\big| \partial_\alpha \langle M \rangle_{\rho_\alpha} \big|^2} \geqslant \frac{1}{I_c(\boldsymbol{\Lambda}, \rho_\alpha)} \geqslant \frac{1}{I_q(\rho_\alpha)},
    \end{equation}
    where we have used the shorthand $\partial_\alpha=\frac{\partial}{\partial \alpha}$.
    Essentially, these inequalities follow from the error propagation formula \cite{toth2014quantum,pezze2018quantum}, and classical and quantum Cramer-Rao bounds \cite{sidhu2020geometric}.
    The central term here is the reciprocal of the classical Fisher information
    \begin{equation}
        \label{eq:cfi}
        I_c(\boldsymbol{\Lambda}, \rho_\alpha)=\sum_{i=1}^{2^m} \frac{\left(\partial_\alpha p_i\right)^2}{p_i}
    \end{equation}
    with $p_i = \Tr \Lambda_i \rho_\alpha$ \cite{meyer2021variational}.
    The right-hand side of the inequality contains the quantum Fisher information, which, among other ways \cite{sidhu2020geometric}, can be calculated as
    \begin{equation}
        \label{eq:qfi-fidelity}
        I_q(\ra) = 8\frac{1 - F(\ra, \rho_{\alpha + \mathrm{d}\alpha})}{\mathrm{d}\alpha^2}
    \end{equation}
    with $F(\rho, \tau) = \Tr \sqrt{\sqrt{\rho} \tau \sqrt{\rho}}$ being the fidelity between the states $\rho$ and $\tau$.
    In this work, we use the inequality \eqref{eq:var-qcrb} for assessing the quality of the observable found, e.g., by solving \eqref{eq:ls_min}.

\section{Results}
\label{sec:results}

    In this section, we show our main results.
    First, we derive the optimal observable and its variance for a task of predicting the coefficient of a mixture of quantum states.
    Second, we present three interesting observations regarding regression tasks on pure labeled states.
    We support our analytical results with numerical experiments.

    \subsection{Convex combination of states}
    \label{sec:lin-comb}
    
        Let us consider the task of predicting the label $\alpha \in [0, 1]$ which is encoded into an $n$-qubit state as
        \begin{equation}
            \label{eq:lin_al}
            \rho_\alpha = \alpha\rho^{(1)} + (1 - \alpha)\rho^{(2)},
        \end{equation}
        where we assume
        \begin{equation}\label{eq:app:ex1}
            \rho^{(1)} = r\dyad{v_1} + (1 - r)\dyad{v_2},\quad \rho^{(2)} = \frac{1}{2^n} \Id
        \end{equation}
        with $0 \leqslant r\leqslant 1$, and $\ket{v_1}$ and $\ket{v_2}$ being orthonormal vectors.
        That is, the labeled state $\rho_\alpha$ is a mixture of a state $\rho^{(1)}$ of the rank at most two with the maximally mixed state $\rho^{(2)}$.
        Although this model may look rather simple, it nonetheless captures a number of interesting cases of the application of the method \eqref{eq:ls_min} with different numbers of measured qubits $m$ in \eqref{eq:obs-par}.        
        Namely, we will see that for $1 \leqslant k \leqslant m \leqslant n$ there are observables $M_m$ of the form
        \begin{equation}
            \label{eq:obs-m}
            M_{m} = \sum_{i=1}^{2^m} \lambda_i \Lambda_i
        \end{equation}
        such that they may give the label with comparable accuracy, i.e., $\langle M_k\rangle_{\rho_\alpha} \approx \langle M_m \rangle_{\rho_\alpha}$, but with generally larger \textit{total variance} $\int_{0}^1 \Delta^2_{\rho_\alpha} M_k  d\alpha \geqslant \int_{0}^1 \Delta^2_{\rho_\alpha} M_m d\alpha$.
    
        Let us state the problem more formally.
        Essentially, we want to solve the following minimization task:
        \begin{align}
        \label{eq:glob_task}
            \begin{split}
                &M_m^* \in \arg\min_{M_m}\,\int_0^1\,\Delta_{\rho_\alpha}^2 M_m\,d\alpha \\
                &\mathrm{s.t.} \quad\Tr M_m \rho_\alpha  = \alpha.
            \end{split}
        \end{align}
        Although this problem can be solved analytically for the considered state $\rho_\alpha$, the derivations are rather technical and left in Appendix~\ref{app:mixture}.
        However, here we shall outline the recipe for obtaining our solution.
        
        First, as was done in \cite{kardashin2025predicting, Holevo_2012}, considering a small perturbation about an optimal observable $M = M^* + \epsilon Y$, and applying the Lagrange multipliers method,  we arrive to a Lyapunov equation:
        \begin{equation}
            \label{eq:opt_glob}
            \rho_{1/2} M^* + M^*\rho_{1/2} = \rho_{1/2} - \mu(\rho^{(1)}-\rho^{(2)}),
        \end{equation}
        where $\rho_{1/2}$ is $\rho_\alpha$ taken at the point $\alpha=1/2$, and $\mu$ is a Lagrange multiplier.
        Then, using the notion of the symmetric logarithmic derivative \cite{sidhu2020geometric}, we find that $\mu = -2/I_q(\rho_{1/2})$.
        Putting $M^*$ in the form of the eigendecomposition \eqref{eq:obs-m} into \eqref{eq:opt_glob}, we obtain the eigenvalues $\lambda_i$ expressed through the probabilities $p_i^{(1,2)} = \Tr \Lambda_i \rho^{(1, 2)}$.
        This allows to reduce the problem \eqref{eq:glob_task} of minimization over observables $M_m$ to the maximization of an $f$-divergence \cite{sason2016fdiv} between the probability distributions $p^{(1)} = \{p^{(1)}_i\}_{i=1}^{2^m}$ and $p^{(2)} = \{p^{(2)}_i\}_{i=1}^{2^m}$.
        Finally, we prove that this $f$-divergence is maximized on $\Lambda_i = \ketbra{v_i}$, the eigenprojectors of $\rho^{(1)}$ sorted in descending order of the eigenvalues.
        Therefore, we can write explicit formulas for $p_i^{(1,2)}$ and hence for the eigenvalues $\lambda_i$.
        When we measure all $m=n$ qubits, the optimal observable $M_m^*$ has the eigenvalues
        \begin{align}
            \label{eq:lambdas_opt_n}
            \begin{split}
                \lambda_1 &= \frac12 + \frac{2}{I_q(\rho_{1/2})} \frac{r - 2^{-n}}{r + 2^{-n}}, \\
                \lambda_2 &= \frac12 + \frac{2}{I_q(\rho_{1/2})} \frac{(1 - r) - 2^{-n}}{(1 - r) + 2^{-n}}, \\
                \lambda_{i\geqslant 3} &= \frac12 - \frac{2}{I_q(\rho_{1/2})},
            \end{split}
        \end{align}
        where the quantum Fisher information is
        \begin{equation}
            \label{eq:licomb_qfi}
            I_q(\rho_\alpha) = \frac{\alpha DE -2 r (1 - D) + 2^n - 1}{(1 - \alpha) (1 - \alpha D ) (1 - \alpha E )}.
        \end{equation}
        with $D = 1 - 2^n (1 + r)$ and $E = 1 - 2^n r$.
        If we measure $m<n$ qubits, the $2^{n-m}$-degenerate eigenvalues of $M_m^*$ have the form
        \begin{align}
            \label{eq:lambdas_opt_m}
            \begin{split}
                &\lambda_{1 \leqslant i\leqslant 2^{n-m}} = 1, \\
                &\lambda_{i > 2^{n-m}} = \frac{1}{1-2^m},
            \end{split}
        \end{align}
        where we have no dependence on $r$.
        
        Having found the eigenvalues $\lambda_i$ and eigenprojectors $\Lambda_i$, we therefore have found an optimal observable $M^*_m$.
        As noted above, \eqref{eq:opt_glob} is a Lyapunov equation of the form $AX + XA = B$.
        In \cite{Personick71}, it is proven that the solution $X$ is Hermitian and unique as long as $A$ is strictly positive-definite.
        Since $A$ in our case is defined by \eqref{eq:lin_al} and \eqref{eq:app:ex1}, the observable $M^*_m$ we found is unique unless $\alpha=1$.

        Recall that $M_m^*$ gives the label $\alpha$ of $\rho_\alpha$ in expectation, i.e., $\langle M_m^* \rangle_{\rho_\alpha} = \alpha$.
        For us, important also is the variance of $M_m^*$ in the state $\rho_\alpha$.
        When we measure all the qubits, $m=n$, this variance is
        \begin{multline}
            \label{eq:obs_opt-n}
            \Delta^2_{\rho_\alpha} M_n^* = (1 - \alpha) \alpha + \frac{(2 \alpha - 1) (1 - 2^n A) A}{B^2} \\ + \frac{2 (2 + 2^n) C - a \big(1 + 2 (4 + 2^n) C\big)}{B},
        \end{multline}
        where
        \begin{align*}
            A &= (1 - 2 r)^2, \\
            B &= 1 - 2^n + 2^n (2^n - 4) (r - 1) r, \\
            C &= r(r - 1).
        \end{align*}
        At the same time, when only a fraction $m < n$ qubits is measured, the variance becomes
        \begin{equation}
            \label{eq:obs_opt-m}
            \Delta^2_{\rho_\alpha} M_m^* = (1-\alpha)\left( \frac{1}{2^m - 1} + \alpha \right),
        \end{equation}
        with, again, no dependence on $r$.
        
        We remind that the derivation of the formulas in this Section is given in Appendix~\ref{app:mixture}.

        \subsubsection{Number of measured qubits $m$}   

            Let us look more closely at the optimal observable $M_m^*$ and its eigenvalues.
            One can notice that the degeneracy of the eigenvalues is dependent on $r$.
            Indeed, putting $r = 1/2$ into \eqref{eq:lambdas_opt_n}, one obtains a $2$-fold degenerate eigenvalue $\lambda_1 = \lambda_2 = 1$ and a $(2^n-2)$-fold degenerate  $\lambda_3 = \frac{2}{2^n-2}$. 
            In this case, the corresponding observable $M_m^*$ can be constructed from $2$-dimensional projectors, since $(2^n-2)$-dimensional eigenspace for $\lambda_3$ can be split into projectors of this type.
            Evaluation of this observable can be realized by measuring no less than $m = n-1$ qubits of the state $\rho_\alpha$.
            However, since there are only two distinct eigenvalues, applying the Naimark's extension \eqref{eq:naimark-variational}, it is sufficient to introduce only $m_{\mathrm{a}}=1$ auxiliary qubit.

            In the case $r=0$  there is a $(2^n-1)$-fold degenerate $\lambda_1 = \lambda_{i\geqslant 3} = \frac{1}{1-2^n}$, and non-degenerate $\lambda_2=1$. 
            The picture is the same for $r=1$ but with the roles of $\lambda_1$ and $\lambda_2$ exchanged.
            The dimensionalities of the eigenspaces are $1$ and $2^n-1$, and hence the optimal observable can be accessed only by measuring all $m=n$ qubits.
            Alternatively, one can use Naimark's extension with measuring again $m_{\mathrm{a}}=1$ additional qubit. 
            
            When $r\notin\{0, 1/2, 1\}$, there are three different eigenvalues: unique $\lambda_1$ and $\lambda_2$, and $(2^n-2)$-fold degenerate $\lambda_{i\geqslant 3}$.
            To avoid measuring all $m=n$ qubits, one can use the Naimark's extension with introducing $m_{\mathrm{a}}=2$ additional qubits.

        \subsubsection{Numerical experiments}   
        \label{sec:lin_comb-numerics}

            Let us now numerically test our analytical results.
            We consider a state $\rho_\alpha$ of $n=5$ qubits of the form \eqref{eq:lin_al} with
            \begin{equation}
                \label{eq:lincomb_ghz}
                \rho^{(1)} = r\ketbra{\mathrm{GHZ_+}} + (1 - r)\ketbra{\mathrm{GHZ_-}},
            \end{equation}
            where
            \begin{equation}
                \ket{\mathrm{GHZ_\pm}} = \frac{1}{\sqrt{2}}\left( \ket{0}^{\otimes n} \pm \ket{1}^{\otimes n} \right),
            \end{equation}
            and we set $r=1/4$.
            We numerically solve the minimization problem \eqref{eq:ls_min} with the weights $\wls=1$ and $\wvar=10^{-4}$ using the BFGS optimizer \cite{wright1999numerical} built into SciPy \cite{2020SciPy-NMeth}.
            As a variational ansatz $U_{\boldsymbol{\theta}}$ in \eqref{eq:obs-par}, we employ a hardware-efficient ansatz \cite{kandala_hardware-efficient_2017} of $l=5$ layers described in Appendix~\ref{app:hea}.
            The training set $\mathcal{T}=\{(\rho_{\alpha_i}, \alpha_i)\}_{i=1}^{10}$ consisted of ten states $\rho_{\alpha_i}$ with equidistant $\alpha_i \in [0,1]$.
    
            In Fig.~\ref{fig:lin_comb}, we plot the error between the predicted $\mathsf{a} = \Tr M_m^* \rho_\alpha$ and true label $\alpha$ for the optimized observable $M_m^*$ and numbers of measured qubits $m\in \{ 1, 3, 5 \}$.
            In addition, we show the analytical variances \eqref{eq:obs_opt-n} and \eqref{eq:obs_opt-m}, as well as the right-hand side of the bound \eqref{eq:var-qcrb} with the quantum Fisher information \eqref{eq:licomb_qfi}.
            As one can see, the prediction error is very small for all $m$.
            However, as anticipated, the variance of the numerically optimized observable $M_m^*$ increases with decreasing $m$.
            For $m=n$, the variance approaches closely the right-hand side of the bound \eqref{eq:var-qcrb}, but does not saturate it for $\alpha$ around $0$; this is in agreement with the analytical variance \eqref{eq:obs_opt-n} also plotted in the figure.

            We note that one generally cannot directly compare the variance of an observable $M$ with the bound \eqref{eq:var-qcrb}.
            Indeed, this bound also includes the derivative of the expectation characterizing the prediction bias \cite{farokhi2025sample}, i.e., $\partial_\alpha\langle M \rangle_{\rho_\alpha} = 1 + \partial_\alpha b_\alpha$.
            However, in our numerical experiments, we obtained observables with $|\partial_\alpha\langle M_m^* \rangle_{\rho_\alpha}|^2\approx1$, which allows us to make the mentioned above comparison.
    
            \begin{figure*}
                \centering
                \includegraphics[width=0.4375\linewidth]{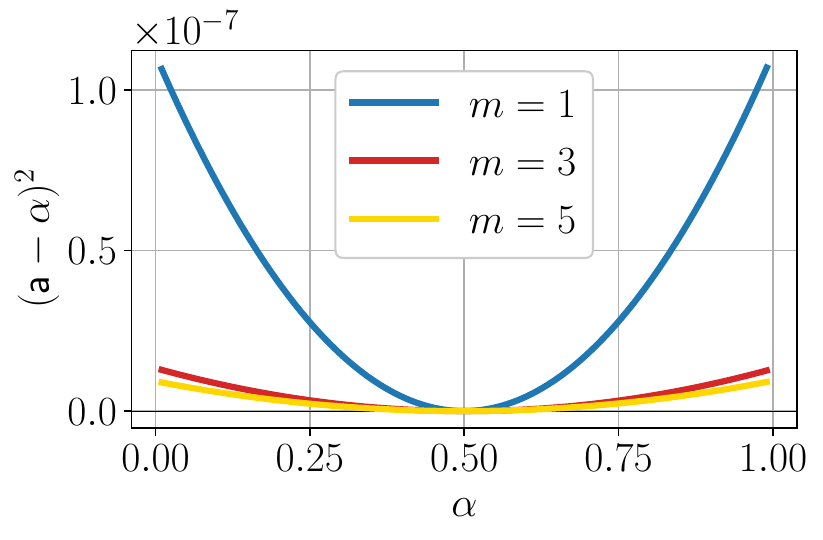}
                \hspace{2em}
                \includegraphics[width=0.45\linewidth]{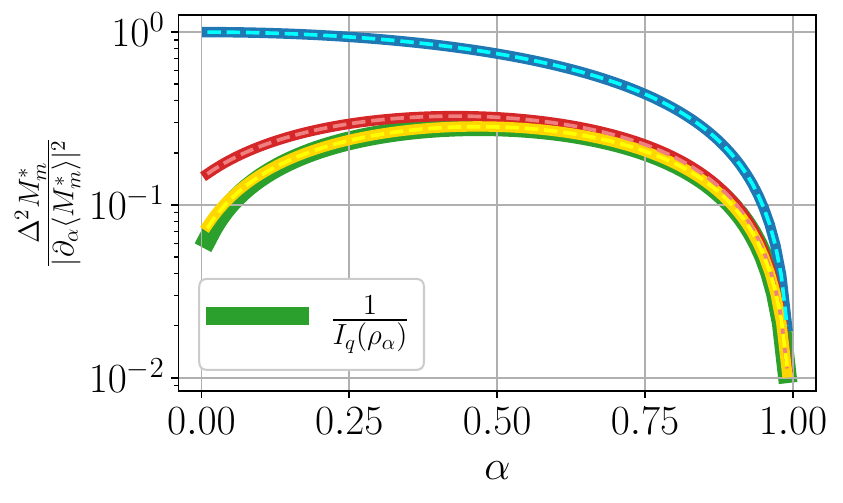}
                \caption{
                    Squared difference between the prediction $\mathsf{a}=\langle M_m^*\rangle_{\rho_\alpha}$ and the true parameter $\alpha$
                    (left) and the variance of the optimized observable $M_m^*$ (right) obtained via numerically solving \eqref{eq:ls_min}.
                    The training set is $\mathcal{T}=\{(\rho_{\alpha_i}, \alpha_i)\}_{i=1}^{10}$ with $\rho_\alpha$ being a state of $n=5$ qubits defined by \eqref{eq:lin_al} and \eqref{eq:lincomb_ghz}, and $\alpha$ are picked equidistantly in $[0, 1]$.
                    Different colors indicate different numbers of measured qubits $m \in \{1,3,5\}$ in \eqref{eq:obs-par}.
                    The parametrized unitary $U_{\boldsymbol{\theta}}$ is represented by HEA described in Appendix~\ref{app:hea}.
                    In the right panel, the dashed lines show the analytical variances \eqref{eq:obs_opt-n} and \eqref{eq:obs_opt-m}, while the solid green line stands for the right-hand side of the bound \eqref{eq:var-qcrb}.
                }
                \label{fig:lin_comb}
            \end{figure*}

\subsection{Pure states}
\label{sec:pure_est}

    Consider a family $\{\ket{\psi_\alpha}\}_\alpha$ of \emph{pure} quantum states parametrized by $\alpha$.
    Let us assume that all $\psia$ belong to some  \emph{real} $d$-dimensional subspace. By that, we mean that each vector $\psia$ can  be expressed as a linear combination of some fixed~(not depending on $\alpha$) orthonormal vectors $\{\ket{c_j}\}_{j=1}^d$:
    \begin{equation}
        \psia = \sum_{j=1}^d\,a_j(\alpha)\ket{c_j},
    \end{equation}
    with~(complex) coefficients $a_j$ with phases not dependent on $\alpha$, i.e., $a_j(\alpha) = \abs{a_j(\alpha)}e^{i\phi_j}$ with $\phi_j$ constant. Under this condition, the coefficients $a_j$ can be made real by redefining each vector $\ket{c_j}$ as $e^{i\phi_j}\ket{c_j}$.
    As a consequence, differentiating the normalization condition $\langle\psi_\alpha|\psi_\alpha\rangle = 1$, we obtain an orthogonality condition
    \begin{equation}
        \langle\psi_\alpha|\partial_\alpha\psi_\alpha\rangle = 0,
    \end{equation}
    where $\ket{\partial_\alpha\psi_\alpha}$ denotes the derivative of $\ket{\psi_\alpha}$ with respect to $\alpha$.
    This allows to simplify a standard expression for the quantum Fisher information for pure states \cite{liu2020quantum}:
    \begin{align}
        I_q(\psi_\alpha) 
        &= 4 \left( \langle\partial_\alpha\psi_\alpha|\partial_\alpha\psi_\alpha\rangle + |\langle\psi_\alpha|\partial_\alpha\psi_\alpha\rangle|^2 \right) \nonumber \\
        &= 4 \langle\partial_\alpha\psi_\alpha|\partial_\alpha\psi_\alpha\rangle.
    \end{align}

    Now let us consider a new orthonormal system of vectors
    \begin{equation}\label{eq:unit_transf}
        \ket{v_i} = \sum_{j=1}^d \,u_{ij}\ket{c_j},
    \end{equation}
    which are connected with $\{\ket{c_i}\}_{i=1}^d$ by a unitary transformation $u$.
    With the measurement probabilities in the new basis
    \begin{equation}
        p_i = \big|\langle v_i|\psi_\alpha\rangle\big|^2
    \end{equation}
    one can obtain the classical Fisher information:
    \begin{multline}
        I_c(\{v_i\}_{i=1}^d,\psi_\alpha) = \sum_{i=1}^d\,\frac{(\partial_\alpha p_i)^2}{p_i}  \\
        = \sum_{i=1}^d\,\Bigg[2\big|\langle v_i|\partial_\alpha\psi_\alpha\rangle\big|^2 \hfill \\ \hfill
        + \frac2{\big|\langle v_i|\psi_\alpha\rangle\big|^2} \Re{\langle v_i|\partial_\alpha\psi_\alpha\rangle^2\langle\psi_\alpha|v_i\rangle^2}\Bigg].
    \end{multline}
    Using the fact that the real part is less than the absolute value, we can upper bound this expression as
    \begin{multline}
        I_c(\{v_i\}_{i=1}^d, \psi_\alpha)\leqslant   \sum_{i=1}^d\,\Bigg[2\big|\langle v_i|\partial_\alpha\psi_\alpha\rangle\big|^2  \hfill \\ \hfill
        + \frac2{\big|\langle v_i|\psi_\alpha\rangle\big|^2} \big|\langle v_i|\partial_\alpha\psi_\alpha\rangle\big|^2 \big|\langle\psi_\alpha|v_i\rangle\big|^2\Bigg] \\
        = 4\sum_{i=1}^d\,\big|\langle v_i|\partial_\alpha\psi_\alpha\rangle\big|^2 = 4\langle\partial_\alpha\psi_\alpha|\partial_\alpha\psi_\alpha\rangle = I_q(\alpha).
    \end{multline}
    The upper bound, which is the quantum Fisher information,  can be attained if we choose the transformation $u$ in~(\ref{eq:unit_transf}) to be real (and hence orthogonal). 
    In this case, we refer to the  basis $\ket{v_i}$ as \emph{real}.
    We come to the first conclusion in this section. 
    \begin{observation}\label{app:fact1}
        If $\{\ket{\psi_\alpha}\}_\alpha$ is a parametrized family of pure states in a real $d$-dimensional subspace, then an observable $M$, which has the eigenprojectors $\{ \ketbra{v_i} \}_{i=1}^d$ with elements from a real basis $\{ \ket{v_i} \}_{i=1}^d$, would produce $I_c(\{ v_i \}_{i=1}^d, \psi_\alpha) = I_q(\psi_\alpha)$.
    \end{observation}

    Recall that this observable $M$, besides the projectors, also depends  on its eigenvalues $\lambda_i$:
    \begin{equation}
        \label{eq:proj_dec}
        M = \sum_{i=1}^d\, \lambda_i\dyad{v_i}.
    \end{equation}
    The variance of this observable in a state $\psia$ is
    \begin{equation}
        \Delta_{\psi_\alpha}^2 M 
        = \sum_{i=1}^d\,\lambda_i^2p_i - \left(\sum_{i=1}^d\,\lambda_ip_i\right)^2,
    \end{equation}
    with the measurement probabilities $ p_i = \abs{\langle v_i|\psi_\alpha\rangle}^2$.
    Suppose that $d=2$, i.e., there are only two projectors.
    This results in $\partial_\alpha p_1 = -\partial_\alpha p_2$.
    Therefore, the left-hand side of inequality \eqref{eq:var-qcrb} becomes
    \begin{align}
        \frac{\Delta_{\psi_\alpha}^2 M}{|\partial_\alpha\langle M\rangle_{\psi_\alpha}|^2}
        &= \frac{\lambda_1^2p_1 + \lambda_2^2p_2 - \lambda_1^2p_1^2 - \lambda_2^2p_2^2 - 2\lambda_1\lambda_2p_1p_2}{(\partial_\alpha p_1)^2(\lambda_1-\lambda_2)^2} \nonumber \\
        &= \frac{p_1p_2}{(\partial_\alpha p_1)^2}
        = \left(\frac{(\partial_\alpha p_1)^2}{p_1} + \frac{(\partial_\alpha p_2)^2}{p_2}\right)^{-1} \nonumber \\
        &= \frac1{I_c(\{v_i\}_{i=1}^2,\,\psi_\alpha)},
    \end{align}
    with $\lambda_1$ and $\lambda_2$  completely eliminated.
    We come to our second conclusion.
    \begin{observation} 
        \label{app:fact2}
    For a parametrized family of pure states $\{\ket{\psi_\alpha}\}_\alpha$, if the observable $M$ has only two terms in \eqref{eq:proj_dec}, it always gives $\Delta^2_{\psi_\alpha} M/|\partial_\alpha\langle M\rangle_{\psi_\alpha}|^2 = 1/I_c(\{v_i\}_{i=1}^2,\,\psi_\alpha)$ whenever $p_1 + p_2 = 1$, i.e., $\ket{\psi_\alpha}\in\mathrm{span}\{\ket{v_1}, \ket{v_2}\}$.
    \end{observation}
    
    Now, let us suppose that a family of parametrized pure states $\rho_\alpha = \dyad{\psi_\alpha}$ belongs to a $2$-dimensional real subspace $V$.
    One can consider a general task of predicting $\alpha$ by the minimization procedure~(\ref{eq:ls_min})\footnote{Alternatively, one may consider finding the optimal observable analytically via, e.g, Eq.~40 in \cite{kardashin2025predicting}.}.
    For the states under consideration, the optimal observable $M^*$ can be taken to be supported on $V$ and represented by a real $2\times2$ matrix. Its eigenvectors  will then constitute a real basis of $V$. According to Observation~\ref{app:fact2}, the observable will saturate the inverse classical Fisher information at each point $\alpha$. The latter, according to Observation~\ref{app:fact1}, will saturate the inverse quantum Fisher information at each point. Our third conclusion is then the following.
    \begin{observation}
        \label{app:fact3}
        For a parametrized family of pure states $\{\ket{\psi_\alpha}\}_\alpha$ from a real two-dimensional subspace, one can always find an observable $M$ giving $\Delta^2_{\psi_\alpha} M/|\partial_\alpha\langle M\rangle_{\psi_\alpha}|^2 = 1/I_q(\psi_\alpha)$.
    \end{observation}

    \subsubsection{Numerical experiments}
    \label{sec:ising-numerics}

        If an $n$-qubit state $\ket{\psi_\alpha}$ lives in a real subspace $V$, the above Observations imply certain achievable efficiency of predicting $\alpha$ depending on $\mathrm{dim}\, V$ and the number of measured qubits $m$.
        In this section, we support the Observations with numerical experiments for the transverse field Ising Hamiltonian
        \begin{equation}
            \label{eq:ising_ham}
            H_h = \sum\limits_{i=1}^n \left( Z_i Z_{i+1} + h X_i \right),
        \end{equation}
        where $X_i$ and $Z_i$ are Pauli operators acting on the $i$th qubit, and we apply periodic boundary conditions $Z_{n+1} \equiv Z_1$.
        That is, we consider the task of predicting the transverse field $h$ given the ground state $\ket{\psi_h}$ of $H_h$.
        
        For $n=3$ qubits, this Hamiltonian can be diagonalized using symbolic algebra software, such as SymPy  \cite{10.7717/peerj-cs.103} or Mathematica \cite{Mathematica}.
        This way, one can verify that the (unnormalized) ground state has the form
        \begin{gather*}
            \ket{\psi_h} = \ket{\Psi} + \frac{h - 2 + 2 \sqrt{1 + h(h - 1)}}{3 h} \ket{\Psi^\perp},
        \end{gather*}
        where $\ket{\Psi} = \ket{000}+\ket{111}$.
        That is, the ground state is real and belongs to a two-dimensional subspace $V= \mathrm{span}\{\ket{\Psi}, \ket{\Psi^\perp}\}$.
        Therefore, there must be an optimal observable of the form \eqref{eq:obs-m} with $m=1$ measured qubit for which all the three Observations hold.

        \begin{figure*}
            \centering
            \includegraphics[width=0.3315\linewidth]{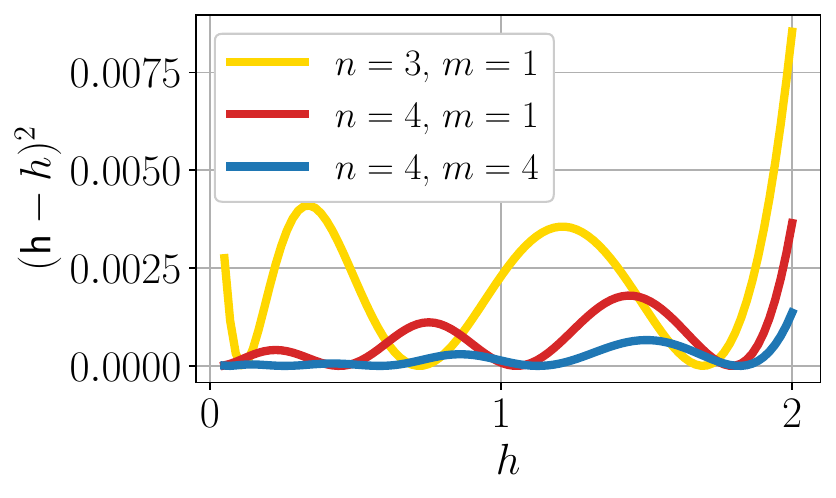}
            \includegraphics[width=0.325\linewidth]{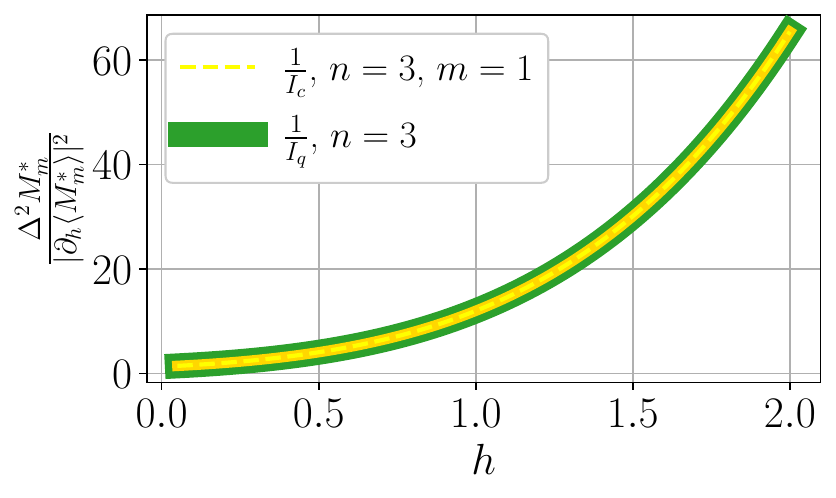}
            \includegraphics[width=0.325\linewidth]{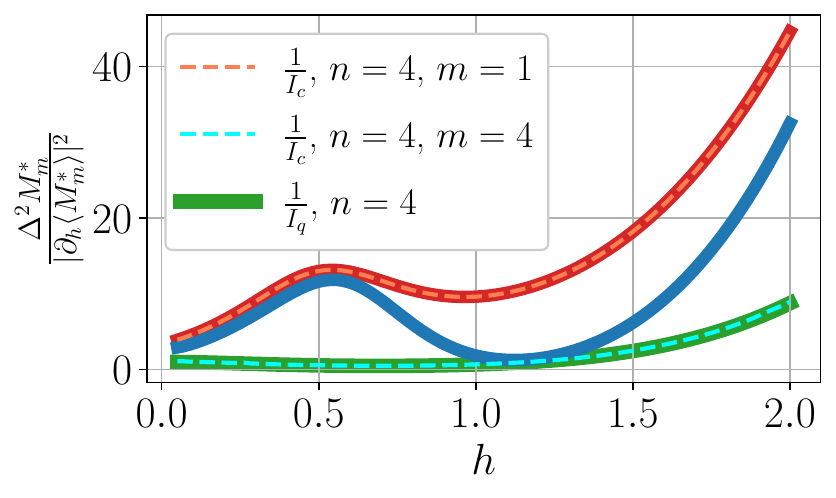}
            \caption{
                Numerical results for predicting the transverse field $h$ of the Ising Hamiltonian \eqref{fig:ising} of $n\in\{3,4\}$ qubits.
                The observable $M_m^*$ is obtained via numerically solving \eqref{eq:ls_min}.
                Left: Squared difference between the prediction $\mathsf{h}=\langle M_m^*\rangle_{\psi_h}$ and true $h$.
                Center: Variance of $M_m^*$ for the case of $n=3$ qubits with $m=1$ qubit measured.
                Right: Variance of $M_m^*$ for the case of $n=4$ qubits with $m \in \{1,4\}$ measured qubits.
                The training set is $\mathcal{T}=\big\{\big(|\psi_{h_i}\rangle, h_i\big)\big\}_{i=1}^{10}$ with $h$ picked equidistantly from $[0.05, 2]$.
                Different colors indicate different numbers of measured qubits $m$ in \eqref{eq:obs-par}, as well as the number of qubits $n$ of the ground state.
                The dashed lines of the corresponding colors show the central part of the the bound \eqref{eq:var-qcrb}, while the solid green line stands for right-hand side of it.
            }
            \label{app:fig:ising-3-4-qubits}
        \end{figure*}

        In the case of $n=4$ qubits, the ground space is not anymore two-dimensional, but it still belongs to a real subspace.
        Therefore, we can expect that with $m=1$ measured qubit we can saturate, by Observation~\ref{app:fact2}, the first inequality in \eqref{eq:var-qcrb}, but not the second.
       If we measure all the qubits, $m=n=4$, Observation~\ref{app:fact1} implies that the classical Fisher information $I_c$  may saturate the second inequality in \eqref{eq:var-qcrb}, but the observable itself need not saturate the first inequality in~\eqref{eq:var-qcrb}.

        To test the claims above, we generate a training set $\mathcal{T}=\big\{\big(|\psi_{h_i}\rangle, h_i\big)\big\}_{i=1}^{10}$ with equidistant $h_i \in [0.05, 2]$ and numerically solve \eqref{eq:ls_min}.
        Recall that with this we intend to train an observable $M$ such that the expectation $\langle M\rangle_{\psi_h} = h + b_h$ has small prediction bias $b_h$ and presumably low variance $\Delta^2_{\psi_h} M$.
        To represent the unitary $U_{\boldsymbol{\theta}}$, we again use HEA described in Appendix~\ref{app:hea};
        for a system of $n=3$ qubits we used $l=2$ layers of the ansatz, and $n=4$ qubits we used $l=4$ layers.

        The results plotted in Fig.~\ref{app:fig:ising-3-4-qubits} show that when the Hamiltonian \eqref{eq:ising_ham} is of $n=3$ qubits, it is enough to measure only $m=1$ qubit for saturating the both inequalities in \eqref{eq:var-qcrb}.
        This means that the found observable $M^*_{m=1}$ has the lowest  possible variance (adjusted by $|\partial_h\langle M^*_{m=1}\rangle_{\psi_h}|^2$).
        This agrees with Observation~\ref{app:fact3}, and hence \ref{app:fact1} and \ref{app:fact2}.

        For the Ising Hamiltonian of $n=4$ qubits, in Fig.~\ref{app:fig:ising-3-4-qubits} we see that measuring $m=1$ qubit allows to find an observable $M^*_{m=1}$ saturating only the first inequality in \eqref{eq:var-qcrb}, as predicted by Observation~\ref{app:fact2}.
        If all $m=n=4$ qubits are measured, then  $I_c$ saturates   the second inequality in \eqref{eq:var-qcrb}, which is in agreement with Observation~\ref{app:fact1}. At the same time,
        the variance of $M^*_{m=4}$ lies \emph{above} the lower bound $1/I_q(\psi_\alpha)$, but it is still lower than the variance of $M^*_{m=1}$.
        One can also notice that the prediction $\mathsf{h} = \langle M^*_{m}\rangle_{\psi_h}$ is more accurate with $m=4$ than with $m=1$.

    \subsection{Additional numerical results}

        In previous sections, we observed that measuring fewer qubits $m$ may result not only in greater error of label prediction, but also in larger variance of it.
        In Appendix~\ref{app:additional_numerics}, we further support these observations with numerical experiments of label prediction for ground states $\ket{\psi_\alpha}$ of parametrized Hamiltonians $H_\alpha$.
        That is, we solve this prediction task for the Schwinger Hamiltonian (Appendix~\ref{app:schwinger-numerics}), and the cluster Hamiltonian (Appendix~\ref{app:cluster-numerics}).
        These Hamiltonians were studied previously in \cite{nagano2023quantum} and \cite{umeano2023can} in the context of classification of the ground states. 
        Additionally, we again consider the Ising Hamiltonian, but with a greater number of qubits (Appendix~\ref{app:sec:ising-numerics}).

        Recall that we obtain the observables for prediction by numerically solving \eqref{eq:ls_min}.
        For the mentioned above regression problems, we consider various ans\"atze for representing the parametrized unitary $U_{\boldsymbol{\theta}}$.
        Alongside with HEA used earlier, we apply QCNNs, as well as the Hamiltonian variational ansatz \cite{wiersema2020exploring}; these ans\"atze are described in Appendix~\ref{app:sec:ansatzes}.

        Overall, the numerical results in Appendix~\ref{app:additional_numerics} support the claim that the optimal observables $M^*_m$ of the form \eqref{eq:obs-m} predict the labels $\alpha$ of labeled states $\ket{\psi_\alpha}$ with larger variance with fewer measured qubits $m$.

\section{Conclusion}
\label{sec:conclusion}

    When solving a QML task, one may process labeled quantum states with a trained variational ansatz, and then measure a local observable the expectation of which is used for label prediction. 
    Indeed, for instance, one of the distinct features of QCNNs is that one commonly measures a few qubits of the transformed state.
    That is, measured is a local observable having a few distinct eigenvalues.
    In our work, we showed that one may need observables with more eigenvalues for predicting the labels of labeled states with lower variance.

    First, in Section~\ref{sec:lin-comb} we considered a task of predicting the coefficient of a mixture of two quantum states.
    For this task, we have analytically found an optimal observable such that it gives the mixture coefficient in expectation with the minimal possible variance.
    We have also shown that depending on the structure of the states in the mixture, one may need to measure different numbers of qubits for achieving this minimal variance.
    This number of measured qubits is connected to, essentially, the dimensionalities of the eigenprojectors of the optimal observable.
    Therefore, to decrease the number of qubits to be measured, one can employ Naimark's extension.

    Later, in Section~\ref{sec:pure_est} we studied the task of regression on labeled states which are pure.
    We have derived three interesting observations about achieving the lowest possible variance for such labeled states.
    Particularly, we showed that if a pure state lives in a real two-dimensional subspace, one can always find an observable with the variance equal to the reciprocal of quantum Fisher information, saturating therefore the inequalities \eqref{eq:var-qcrb}.

    Finally, we considered regression problems of predicting the parameter of a parametrized Hamiltonian given its ground state.
    Our numerical experiments confirm the claim that the more qubits one measures (i.e., the more distinct eigenvalues one has in the observable measured), the lower label prediction variance one may get.
    The results of these experiments can be found in Appendix~\ref{app:additional_numerics}.
    
    We emphasize that the origin of higher label prediction variance may be not only the circuit architecture used for processing the labeled states, but also the (local) observable measured after its application.
    Concretely, different regression tasks may require observables with different dimensionalities of its eigenprojectors.
    If one uses a variational ansatz for transforming the observable, this results in different numbers of qubits one needs to measure, which defines the observable's locality.
    While measuring local observables may be beneficial in variational algorithms \cite{cerezo2025does}, this, as we have shown, may also result in larger prediction variance in regression tasks.

\onecolumngrid

\newpage

\appendix

\section{Predicting the weight in a mixture of states}
\label{app:mixture}

    In this Section, we derive the equations shown in Section~\ref{sec:lin-comb} 
    
    \subsection{The connection between global and local optimization tasks}
     Given two fixed states $\rho^{(1)}$ and $\rho^{(2)}$,  consider a task of predicting the parameter $\alpha \in [0, 1]$ of a density operator 
     \begin{equation}
      \label{eq:lin_al-app}   \rho_{\alpha} = \alpha\rho^{(1)} + (1-\alpha)\rho^{(2)}.
     \end{equation}

     The goal is to find an observable $M$ with the average $\Tr\rho_{\alpha}M = \alpha$ and minimal total variance 
     \begin{equation}
         \int_0^1\,\Delta_{\rho_\alpha}^2 M\,d\alpha,
     \end{equation}
    where
    \begin{equation}\label{eq:varata}
        \Delta_{\rho_\alpha}^2 M = \Tr\rho_{\alpha}M^2 - (\Tr\rho_{\alpha}M)^2.
    \end{equation}
    Formally, this is an optimization task 
    \begin{equation}\label{eq:glob_task-app}
        \begin{split}
            & M^*\in \arg\min_{M}\,\int_0^1\,\Delta_{\rho_\alpha}^2 M\,d\alpha \\
            & \mathrm{s. t.}\quad\Tr\rho^{(1)} M = 1,\quad \Tr\rho^{(2)} M = 0
        \end{split}
    \end{equation}
    with the constraints originating from the equality $\Tr\rho_{\alpha} M = \alpha$.
    
    Let us first consider a local version of the task, namely minimization of variance at a given point $\alpha$:
    \begin{equation}
        \begin{split}
            & M^*\in \arg\min_{M}\,\Delta_{\rho_\alpha}^2 M \\ 
            & \mathrm{s. t.}\quad\Tr\rho^{(1)} M = 1,\quad \Tr\rho^{(2)} M = 0.
        \end{split}
    \end{equation}
    It can be approached with the method of Lagrange multipliers by performing minimization of the  functional
    \begin{equation}\label{eq:Lag_fun}
        F[M,\,\mu,\,\nu] = \Delta_{\rho_\alpha}^2 M + \mu(\Tr\rho^{(1)} M - 1) + \nu\Tr\rho^{(2)} M.
    \end{equation}
    The optimal observable $M^*$ is found by considering a small perturbation $\epsilon$  by an arbitrary Hermitian operator $Y$, $M=M^*+ \epsilon Y$ and plugging it into~(\ref{eq:Lag_fun}). Gathering the terms in front of $\epsilon$ and setting them to zero yields
    \begin{equation}
        \Tr{(\rho_{\alpha}M^* + M^*\rho_{\alpha})Y} - 2\Tr{\rho_{\alpha}M^*}\Tr\rho_{\alpha}Y + \mu_0\Tr\rho^{(1)} Y + \nu_0\Tr\rho^{(2)} Y = 0,
    \end{equation}
    where $\mu_0$ and $\nu_0$ are the optimal values of the Lagrange multipliers $\mu$ and $\nu$, respectively.
    The last equation holds for any Hermitian $Y$, hence the sum of the terms in front of $Y$ can be set to zero:
    \begin{equation}\label{eq:int_op_var}
       \rho_{\alpha}M^* + M^*\rho_{\alpha} - 2\Tr{\rho_{\alpha}M^*}\rho_{\alpha} + \mu_0\rho^{(1)} + \nu_0\rho^{(2)} = 0.
    \end{equation}
    Taking the trace of both parts of (\ref{eq:int_op_var}) and using $\Tr\rho_{\alpha} M^* = \alpha$, one obtains
    \begin{equation}\label{eq:munu}
        \mu_0 + \nu_0 = 0.
    \end{equation}
    After  making use of~(\ref{eq:munu}), equation~(\ref{eq:int_op_var}) takes the form:
    \begin{equation}\label{eq:loc_theory-app}
        \rho_{\alpha}M^* + M^*\rho_{\alpha} = 2\alpha\rho_{\alpha} - \mu_0(\rho^{(1)} - \rho^{(2)}).
    \end{equation}
    Multiplying both parts of~(\ref{eq:loc_theory-app}) by $M^*$ and taking the trace yields:
    \begin{equation}\label{eq:trsq}
        \Tr\rho_{\alpha} (M^*)^2 = -\frac{\mu_0}{2} + \alpha^2,
    \end{equation}
    and hence the connection of variance with the optimal value of the Lagrange multiplier $\mu$ is
    \begin{equation}\label{eq:var_con}
    \Delta_{\rho_\alpha}^2 M^* = -\frac{\mu_0}{2}.
    \end{equation}

    Equation~(\ref{eq:loc_theory-app}) can be solved with a well-known ansatz \cite{sidhu2020geometric} 
    \begin{equation}\label{eq:H_ans}
        M^* = \alpha\Id - \frac{\mu_0}{2}L(\alpha),
    \end{equation}
    where $L$ is the symmetric logarithmic derivative~(SLD) operator satisfying the equation
    \begin{equation}
       \label{eq:sld} \frac12\left(\rho_{\alpha}L + L\rho_{\alpha} \right) = \partial_{\alpha}\rho_{\alpha}.
    \end{equation}
    Substitution of~(\ref{eq:H_ans}) into~(\ref{eq:trsq}) with noting that $\partial_{\alpha}\rho_{\alpha}  = \rho^{(1)} - \rho^{(2)}$ gives the connection
    \begin{equation}\label{eq:qFI_con}
        -\frac{\mu_0}{2} = \frac{1}{I_q(\alpha)},
    \end{equation}
    where
    \begin{equation}\label{eq:qFI}
        I_q(\alpha) = \Tr\rho_{\alpha}L^2,
    \end{equation}
    one of the definitions of the quantum Fisher information \cite{sidhu2020geometric}.

    The SLD operator~(and hence $M^*$) can be found directly from~(\ref{eq:sld}), but, for our purposes, let us express the solution in the spectral decomposition form with eigenprojectors $\tilde{\Lambda}_i$ and eigenvalues $\lambda_i$:
    \begin{equation}
        M^* = \sum_i\,\lambda_i\tilde{\Lambda}_i.
    \end{equation}
    Inserting this representation into~(\ref{eq:loc_theory-app}), multiplying both sides by $\tilde{\Lambda}_i$ and taking the trace, we obtain
    \begin{equation}\label{eq:x_expr}
        \lambda_i = \alpha - \frac{\mu_0\left(\tilde p_i^{(1)} - \tilde p_i^{(2)}\right)}{2\left(\alpha \tilde p_i^{(1)} + (1-\alpha)\tilde p_i^{(2)}\right)},
    \end{equation}
    where 
    \begin{equation}\label{eq:ps_def}
        \tilde p_i^{(1)} = \Tr\tilde\Lambda_i\rho^{(1)},\qquad \tilde p_i^{(2)} = \Tr\tilde\Lambda_i\rho^{(2)}
    \end{equation}
    and we assume that $\alpha \tilde p_i^{(1)} + (1-\alpha)\tilde p_i^{(2)}$ is not zero for each $i$, which holds, for example, if $0<\alpha<1$.
    
    Multiplying both parts of~(\ref{eq:x_expr}) by $\tilde p_i^{(1)}$ and summing these expressions over $i$, with the use of~(\ref{eq:var_con}) we arrive at the expression for variance:
    \begin{equation}
        \label{eq:proj_var_l} 
        \Delta_{\rho_\alpha}^2 M^* = \frac{1-\alpha}{\sum_i\,\frac{(\tilde p_i^{(1)})^2 - \tilde p_i^{(1)}\tilde p_i^{(2)}}{\alpha \tilde p_i^{(1)} + (1-\alpha)\tilde p_i^{(2)}}},
        \qquad
        0<\alpha<1.
    \end{equation}
    Such form of the solution might imply that the optimal variance at point $\alpha$ can  be represented as the result  of minimization of the expression 
    \begin{equation}\label{eq:proj_var_l_n}
        \frac{1-\alpha}{\sum_i\,\frac{(p_i^{(1)})^2 - p_i^{(1)}p_i^{(2)}}{\alpha p_i^{(1)} + (1-\alpha)p_i^{(2)}}}
    \end{equation}
    over the set of orthogonal projectors $\Lambda_i$ such that
    \begin{equation}
     p_i^{(1)} = \Tr\Lambda_i\rho^{(1)},\quad  p_i^{(2)} = \Tr\Lambda_i\rho^{(2)}.
    \end{equation}
    This argument can  be supported by comparing the denominator in~(\ref{eq:proj_var_l_n}) with the classical Fisher information for given projectors $\boldsymbol{\Lambda} = \{\Lambda_i\}_i$.
    The latter, calculated via formula \eqref{eq:cfi}, reads
    \begin{equation}\label{eq:app:cFI_ref}
        I_c\left(\boldsymbol{\Lambda}, \rho_\alpha\right) =\sum_i\,\frac{(p_i^{(1)} - p_i^{(2)})^2}{\alpha p_i^{(1)} + (1-\alpha)p_i^{(2)}}. 
    \end{equation}
    Taking the difference between the denominator of~(\ref{eq:proj_var_l_n})(divided by $1-\alpha$) and $I_c$, we have:
    \begin{multline}\label{eq:var_cfi_eq}
        \frac{1}{1-\alpha}\sum_i\,\frac{(p_i^{(1)})^2 - p_i^{(1)}p_i^{(2)}}{\alpha p_i^{(1)} + (1-\alpha)p_i^{(2)}} -\sum_i\,\frac{(p_i^{(1)} - p_i^{(2)})^2}{\alpha p_i^{(1)} + (1-\alpha)p_i^{(2)}} = \sum_i\,\frac{(p_i^{(1)} - p_i^{(2)})(\frac{p_i^{(1)}}{1-\alpha} - (p_i^{(1)} - p_i^{(2)}))}{\alpha p_i^{(1)} + (1-\alpha)p_i^{(2)}}\\=\frac{1}{1-\alpha}\sum_i\,\frac{(p_i^{(1)} - p_i^{(2)})(\alpha p_i^{(1)} + (1-\alpha)p_i^{(2)})}{\alpha p_i^{(1)} + (1-\alpha)p_i^{(2)}} =\frac{1}{1-\alpha}\sum_i\, (p_i^{(1)} - p_i^{(2)}) = 0.
    \end{multline}
    Therefore, the expression in ~(\ref{eq:proj_var_l_n}) is equal to  $1/I_c\left(\boldsymbol{\Lambda}, \rho_\alpha\right)$.  With the use of this fact equation~(\ref{eq:proj_var_l}), in turn,  can be rewritten as 
    \begin{equation}
        \Delta_{\rho_\alpha}^2 M^* = \frac1{I_c(\tilde{\boldsymbol{\Lambda}}, \rho_\alpha)},\quad0<\alpha<1.
    \end{equation}
    By~(\ref{eq:var_con}) and~(\ref{eq:qFI_con}), this expression also equals $1/I_q(\rho_\alpha)$, and so $I_c(\tilde{\boldsymbol{\Lambda}}, \rho_\alpha)$ attains its maximal possible value, $I_q(\alpha)$, on the projectors $\tilde{\boldsymbol{\Lambda}}=\{\tilde\Lambda_i\}_i$, as it should be.

    \begin{remark}
    One can directly set the task of minimizing the variance over $\lambda_i$'s for given fixed distributions $p_i^{(1)}$ and $p_i^{(2)}$, which come from fixed projectors $\{\Lambda_i\}_i$, such that $p_i^{(1)} = \Tr\Lambda_i\rho^{(1)},\,  p_i^{(2)} = \Tr\Lambda_i\rho^{(2)}$. The constraint $\Tr\rho_\alpha M = \alpha$ implies 
    \begin{equation}
        \sum_i\,p_i^{(1)} \lambda_i = 1,
        \qquad
        \sum_i\,p_i^{(2)} \lambda_i = 0,
    \end{equation}
    and the task is formulated as the minimization of the functional
    \begin{equation}
         \Delta_{\rho_\alpha}^2 M +\mu\left(\sum_i\,p_i^{(1)} \lambda_i-1\right) + \nu\sum_i\,p_i^{(2)} \lambda_i,
    \end{equation}
    where, according to~(\ref{eq:varata}),
    \begin{equation}
        \Delta_{\rho_\alpha}^2 M = \sum_i\,\lambda_i^2\left(\alpha p_i^{(1)} + (1-\alpha)p_i^{(2)}\right) - \left(\sum_i\,\lambda_i\left(\alpha p_i^{(1)} + (1-\alpha)p_i^{(2)}\right)\right)^2.
    \end{equation}
    In particular, in such minimization a connection between the optimal value of $\mu$ and  the variance of the optimal observable reads
    \begin{equation}
        \mu_0 = -2\Delta_{\rho_\alpha}^2 M^*,
    \end{equation}
    which has the same form as in~(\ref{eq:var_con}). The solution
    \begin{equation}
        M^* = \sum_i\,\tilde \lambda_i\Lambda_i
    \end{equation} is given by expressions analogous to~(\ref{eq:x_expr}) and~(\ref{eq:proj_var_l}):
    \begin{equation}\label{eq:remsol}
        \tilde \lambda_i = \alpha + \frac{p_i^{(1)} - p_i^{(2)}}{I_c(\boldsymbol{\Lambda}, \rho_\alpha)\left(\alpha p_i^{(1)} + (1-\alpha) p_i^{(2)}\right)},
        \qquad
        \Delta_{\rho_\alpha}^2 M^* = \frac1{I_c(\boldsymbol{\Lambda}, \rho_\alpha)},
        \qquad
        0<\alpha<1,
    \end{equation}
    with $I_c(\tilde{\boldsymbol{\Lambda}}, \rho_\alpha)$ defined by~(\ref{eq:app:cFI_ref}).
    \end{remark}

    Now let us return to the original global task~(\ref{eq:glob_task}). The total variance over the range of the parameter $\alpha$ can be written as
    \begin{equation}    
    \int_0^1\,\Delta_{\rho_\alpha}^2 H\,d\alpha = \Tr\Tilde{\rho}H^2 - \int_0^1\,(\Tr\rho(\alpha)H)^2\,d\alpha,
    \end{equation}
    where
    \begin{equation}\label{eq:rho_til}
        \Tilde{\rho} = \int_0^1\rho_\alpha d\alpha = \frac12\rho^{(1)} + \frac12\rho^{(2)} = \rho_{1/2},
    \end{equation}
    with $\rho_{1/2}$ being $\rho_\alpha$ taken at $\alpha = 1/2$, as in the main text.
    
    With the Lagrange multipliers method, the task~(\ref{eq:glob_task-app}) is reformulated as the minimization of the functional
    \begin{equation}
        \Tr\Tilde{\rho} M^2 - \int_0^1\,(\Tr\rho_\alpha M)^2\,d\alpha + \mu(\Tr\rho^{(1)} M - 1) + \nu\Tr\rho^{(2)} M.   
    \end{equation}
    Considering the same procedure as for~(\ref{eq:Lag_fun}), one arrives at the operator equation for the optimal observable
    \begin{equation}\label{eq:opt_glob_int}
        \Tilde{\rho} M^* + M^*\Tilde{\rho} - 2\int_0^1\Tr{\rho_\alpha M^*}\rho_\alpha d\alpha + \mu_0(\rho^{(1)}-\rho^{(2)}) = 0.
    \end{equation}
    With the integral on the left side being calculated as
    \begin{equation}    2\int_0^1\Tr{\rho_\alpha M^*}\rho_\alpha d\alpha = 2\int_0^1\alpha\rho_\alpha d\alpha = \frac23\rho^{(1)} + \frac13\rho^{(2)} = \rho_{2/3},
    \end{equation}
    and with the use of~(\ref{eq:rho_til}), equation~(\ref{eq:opt_glob_int}) takes form:
    \begin{equation}
     \label{eq:opt_glob_sint}  
     \rho_{1/2} M^* + M^*\rho_{1/2} - \rho_{2/3} + \mu_0(\rho^{(1)}-\rho^{(2)}) = 0.
    \end{equation}
    Multiplying both parts of this equation  by $H_0$ and taking trace results in
    \begin{equation}
        \Tr\Tilde{\rho} (M^*)^2 = -\frac{\mu_0}{2} + \frac13,
    \end{equation}
    and
    \begin{equation}
        \label{eq:app:tot_var_opt}    \int_0^1\,\Delta_{\rho_\alpha}^2 M\,d\alpha = \Tr\Tilde{\rho}M^2 - \frac13 =  -\frac{\mu_0}{2}.
    \end{equation}
    
    Next, we observe that
    \begin{equation}
        - \rho_{2/3}+ \frac16(\rho^{(1)}-\rho^{(2)}) = -\frac12(\rho^{(1)}+\rho^{(2)}) = -\rho_{1/2},
    \end{equation}
    and hence it is convenient to make the substitution $\mu_0 = \Tilde{\mu}_0 + 1/6$, which transforms~(\ref{eq:opt_glob_sint}) into the final form of the equation for the optimal observable:
    \begin{equation}
        \rho_{1/2} M^* + M^*\rho_{1/2} = \rho_{1/2} - \Tilde{\mu}_0(\rho^{(1)}-\rho^{(2)}).
    \end{equation}
    It can be seen that this equation for the global task coincides with the equation~(\ref{eq:loc_theory-app}) for the local task at $\alpha = 1/2$, hence
    making use of the connection~(\ref{eq:qFI_con}) brings the former into the final form of the equation for the \emph{global} optimal observable:
    \begin{equation}
        \label{eq:opt_glob-app}
        \rho_{1/2} M^* + M^* \rho_{1/2} = \rho_{1/2} + \frac{2}{I_q(\rho_{1/2})}(\rho^{(1)}-\rho^{(2)}).
    \end{equation}
    Accordingly, the solution of~(\ref{eq:opt_glob-app}) is given by~(\ref{eq:H_ans}) with substitution $\alpha = 1/2$:
    \begin{equation}
        M^* = \frac12\Id -\frac{\Tilde{\mu}_0}{2}L\left(1/2\right) = \frac12\Id + \frac{1}{I_q(\rho_{1/2})}L\left(1/2\right).
    \end{equation}
    Finally, the optimal total variance is obtained from~(\ref{eq:app:tot_var_opt})
    \begin{equation}
        \int_0^1\,\Delta_{\rho_\alpha}^2 M^*\,d\alpha  =  -\frac{\mu_0}{2} = -\frac{\Tilde{\mu}_0}{2} - \frac1{12} = \frac{1}{I_q(\rho_{1/2})} - \frac1{12}.
    \end{equation}

    We note that Eq.~(\ref{eq:opt_glob-app}) has the structure of a Lyapunov equation~\cite{Bhatia_1997}
    \begin{equation}\label{eq:app:lyap}
        AX + XA = B
    \end{equation}
    with $A$ equal to $\rho_{1/2}$. 
    A useful property proved in Ref.~\cite{Personick71} states that if $A$ is strictly positive, then the solution $X$ must be Hermitian and unique.

\subsection{Reduction to the $f$-divergence optimization}

    In view of~(\ref{eq:proj_var_l})-(\ref{eq:remsol}), the expression for the optimal total variance can be cast into the forms  dependent on the projectors $\Lambda_i$ via $p_i^{(1,2)} = \Tr\Lambda_i\rho^{(1,2)}$, i.e.,
    \begin{equation}
     \label{eq:cent_var}   \int_0^1\,\Delta_{\rho_\alpha}^2 M^*\,d\alpha  = \frac{1}{4\sum_i\,\frac{(p_i^{(1)})^2 - p_i^{(1)}p_i^{(2)}}{p_i^{(1)} + p_i^{(2)}}} - \frac1{12} = \frac{1}{2\sum_i\,\frac{(p_i^{(1)} - p_i^{(2)})^2}{ p_i^{(1)} + p_i^{(2)}}} - \frac1{12}.
    \end{equation}
    In particular, the setting in which only several qubits are measured corresponds to optimization of the expression in~(\ref{eq:cent_var}) over the projectors $\Lambda_i$ of constrained rank.

    Let us take a closer look at the denominator in the rightmost part of~(\ref{eq:cent_var}). If each $p^{(2)}_i$ is not zero, it can be represented as
    \begin{equation}
      \label{eq:div_tr} \sum_i\,\frac{\left(p_i^{(1)} - p_i^{(2)}\right)^2}{ p_i^{(1)} + p_i^{(2)}} = \sum_i\, p_i^{(2)}\frac{\left(p_i^{(1)}/p_i^{(2)} - 1\right)^2}{p_i^{(1)}/p_i^{(2)} + 1} = \sum_i\,p_i^{(2)} f\left(\frac{p_i^{(1)}}{p_i^{(2)}}\right), 
    \end{equation}
    where
    \begin{equation}
        f(x)  = \frac{(x-1)^2}{x+1}
    \end{equation}
    is a convex function on the positive half of the real axis.
    Due to the convexity of $f$, the expression in the rightmost part of~(\ref{eq:div_tr}) can be viewed as the \emph{$f$-divergence}~\cite{sason2016fdiv} $D_f$ between probability distributions $p^{(1)}$ and $p^{(2)}$:
    \begin{equation}\label{eq:fdiv}
        \infdiv{p^{(1)}}{p^{(2)}} = \sum_i\,p_i^{(2)} f\left(\frac{p_i^{(1)}}{p_i^{(2)}}\right). 
    \end{equation}
    Finding the optimal total variance can then be viewed as the task of maximization of $D_f$ over the set of orthogonal projectors.

\subsection{Example with the depolarizing noise}

    Asin the main text, now we consider an $n$-qubit state with representation~(\ref{eq:lin_al-app}) and
    \begin{equation}\label{eq:app:ex_dep}
        \rho^{(1)} = r\dyad{v_1} + (1-r)\dyad{v_2},
        \qquad
        \rho^{(2)} = \frac1{2^n}\Id,
        \qquad
        0\leqslant r\leqslant 1,
    \end{equation}
    where $\ket{v_1}$ and $\ket{v_2}$ are the eigenvectors of $\rho^{(1)}$, and $\rho^{(2)}$ is the maximally mixed state, which corresponds to the depolarizing noise model. 
    
    In order to obtain the  observable with the lowest possible total variance, we substitute into~(\ref{eq:opt_glob-app}) the ansatz
    \begin{equation}\label{eq:app:ans_gen_lc}
        M^* = \sum_{i=1}^{2^n}\,\lambda_i\dyad{v_i}
    \end{equation}
    with  $\{\ket{v_i}\}_i$ being the eigenvectors of $\rho^{(1)}$. Solving  equation~(\ref{eq:opt_glob}) yields the eigenvalues $\lambda_i$ of $M^*$:
    \begin{align}
        \label{eq:app:eigi}
        \begin{split}
            \lambda_1 &= \frac12 + \frac{2}{I_q(\rho_{1/2})} \frac{r - 2^{-n}}{r + 2^{-n}}, \\
            \lambda_2 &= \frac12 + \frac{2}{I_q(\rho_{1/2})} \frac{(1 - r) - 2^{-n}}{(1 - r) + 2^{-n}}, \\
            \lambda_{i\geqslant 3} &= \frac12 - \frac{2}{I_q(\rho_{1/2})},
        \end{split}
    \end{align}
    where $I_q(\rho_{\alpha})$ is the quantum Fisher information of $\rho_\alpha$.
    The calculation via the formulas involving the eigendecomposition of $\rho_\alpha$ and $\partial_\alpha\rho_\alpha$ \cite{liu2020quantum} gives 
    \begin{equation}
        \label{eq:licomb_qfi_app}
        I_q(\rho_\alpha) = \frac{\alpha DE -2 r (1 - D) + 2^n - 1}{(1 - \alpha) (1 - \alpha D ) (1 - \alpha E )},
    \end{equation}
    where $D = 1 - 2^n (1 + r)$ and $E = 1 - 2^n r$. Its value at $\alpha=1/2$ reads
    \begin{equation}\label{eq:app:qfishlc}
        I_q\left(\rho_{1/2}\right) = 4 - \frac{8(r-1)}{2^n(r-1)-1} - \frac{8r}{2^nr+1}.
    \end{equation}
    By  the argument around~(\ref{eq:app:lyap}), strict positivity of $\rho_{1/2}$ guarantees the uniqueness of the solution of~(\ref{eq:opt_glob-app}) given by~(\ref{eq:app:ans_gen_lc})-(\ref{eq:app:qfishlc}).
    The variance of the observable~(\ref{eq:app:ans_gen_lc}) at each point $\alpha$ can be obtained with the use of~(\ref{eq:app:eigi}) and~(\ref{eq:app:qfishlc}):
    \begin{equation}
        \label{eq:app:obs_opt-n}
        \Delta^2_{\rho_\alpha} M^* = (1 - \alpha) \alpha + \frac{(2 \alpha - 1) (1 - 2^n A) A}{B^2} \\ + \frac{2 (2 + 2^n) C - a \big(1 + 2 (4 + 2^n) C\big)}{B},
    \end{equation}
    where
    \begin{equation}
        A = (1 - 2 r)^2, \quad
        B = 1 - 2^n + 2^n (2^n - 4) (r - 1) r, \quad
        C = r(r - 1).
    \end{equation}

    Now suppose that $m<n$  qubits are being measured. 
    In such a setting, the observable of interest is constructed on orthogonal projectors each having the rank $2^{n-m}$. 
    In view of~(\ref{eq:cent_var})-(\ref{eq:fdiv}), finding the  observable with the smallest total variance in this case reduces to optimization of the $f$-divergence over the set of orthogonal projectors with constrained rank. This can be done analytically due to  proportionality of $\rho^{(2)}$ in~(\ref{eq:app:ex_dep}) to the identity operator and the following properties.
    
    Recall that an $n\times n$ real matrix $T$ is called \emph{stochastic} if it has non-negative entries and $\sum_i\,T_{ij} = 1$ for any $j=1,\,\ldots,\,n$ (i.e., each row sums to unity).
    
    \begin{lemma}[\cite{sagawa2020div}]\label{th:div_ineq}
        Let $f$ be a convex function and $p$, $q$, $p'$, $q'\in\mathbb{R}^d$.  Let all the components of $q$, $q'$ be positive. If $p' = Tp$ and $q' = Tq$ for a stochastic matrix $T$, then
        \begin{equation}
            \sum_i\,q'_i\, f\left(\frac{p'_i}{q'_i}\right)\leqslant \sum_i\,q_i \,f\left(\frac{p_i}{q_i}\right).
        \end{equation}
    \end{lemma}
    
    Let $p,\,p'\in\mathbb{R}^d$ and $p_i^{\downarrow}$, $p'^{\downarrow}_i$, $i = 1,\,\ldots,\,d$ are their components arranged in descending order.
    Vector $p'$ is said to be \emph{majorized} by vector $p$, written as $p'\prec p$, if the following inequalities hold
    \begin{equation}\label{app:eq:maj}
        \sum_{i=1}^k\,p'^{\downarrow}_i\leqslant\sum_{i=1}^k\,p^{\downarrow}_i,\quad k=1,\,\ldots,\,d,
    \end{equation}
    with the last inequality~(at $k = d$) holding as equality.

    An $n\times n$ matrix $D$ is called \emph{doubly stochastic} if it is stochastic and, additionally, $\sum_j\,D_{ij} = 1$ for any $i=1,\,\ldots,\,n$ (i.e., each row and column sums to unity).
    
    \begin{theorem}[\cite{Bhatia_1997}]
        Let $p,\,p'\in\mathbb{R}^d$. The following conditions are equivalent:
        \begin{enumerate}
            \item $p'\prec p$.
            \item 
            There exists a doubly stochastic matrix $D$ such that $p' = Dp$.
        \end{enumerate}
    \end{theorem}

    \begin{theorem}[Schur's Theorem~\cite{Bhatia_1997}]\label{th:schur}
        Let $\mathrm{diag}(A)$ denote the vector whose components are the diagonal entries of a Hermitian matrix $A$ and $\lambda(A)$ the vector whose components are the eigenvalues of $A$ specified in any order. Then the two vectors are in majorization relation
        \begin{equation}
            \mathrm{diag}(A)\prec\lambda(A). 
        \end{equation}
    \end{theorem}

    Now let $\{\ket{v_i}\}_i$ be the eigenvectors of $\rho^{(1)}$ corresponding to eigenvalues $\lambda^{\downarrow}_i(\rho^{(1)})$, in descending order.
    As we measure $m<n$ qubits, we need to choose rank $2^{n-m}$ projectors which maximize the $f$-divergence~(\ref{eq:fdiv}). One can see that it is sufficient to choose the following projectors:
    \begin{gather}\label{eq:app:true_proj}
      \Lambda_1 = \sum_{i=1}^{2^{n-m}} \dyad{v_i}, 
      \quad \ldots, \quad
      \Lambda_k = \sum_{i = (k-1)\,2^{n-m}+1}^{k\,2^{n-m}}\dyad{v_i},
      \quad \ldots, \quad 
      \Lambda_{2^m} = \sum_{i=2^n-2^{n-m} +1}^{2^{n}} \dyad{v_{i}}. 
    \end{gather}
    \begin{lemma}
        The $f$-divergence~(\ref{eq:fdiv}) attains its maximum  on the projectors $\{\Lambda_i\}_i$.
    \end{lemma}
\begin{proof}
For a distribution $p'^{(1)}_i = \Tr\Lambda'_i\rho^{(1)}$ originating from any  rank $2^{n-m}$ projectors $\{\Lambda_i'\}_i$, consider the sum in~(\ref{app:eq:maj}):
\begin{equation}\label{app:eq:maj_ch1}
    \sum_{i=1}^k\,p'^{(1)}_i 
    = \sum_{i=1}^k\,\Tr{\Lambda'_i\rho^{(1)}} 
    = \sum_{i=1}^k\,\sum_{j=(i-1)\,2^{n-m}+1}^{i\,2^{n-m}}\,\langle v'_j|\rho^{(1)}|v'_j\rangle 
    = \sum_{j=1}^{k\,2^{n-m}}\,(\rho^{(1)})'_{jj}, 
\end{equation}
where $p'^{(1)}_i$ are supposed to be arranged in descending order and the projectors $\{\Lambda_i'\}_i$ are constructed on orthonormal vectors $\{\ket{v'_j}\}_j$. These vectors define the diagonal elements of the density operator $(\rho^{(1)})'_{jj} \equiv\langle v'_j|\rho^{(1)}|v'_j\rangle$ in the rightmost part of the last equation.  Denoting also $(\rho^{(1)})_{jj} \equiv \langle v_j|\rho^{(1)}|v_j\rangle$, we have
\begin{equation}\label{app:eq:maj_ch2}
    \sum_{j=1}^{k\,2^{n-m}}\,(\rho^{(1)})'_{jj}\leqslant\sum_{j=1}^{k\,2^{n-m}}\,[(\rho^{(1)})'_{jj}]^{\downarrow}\leqslant\sum_{j=1}^{k\,2^{n-m}}\,[(\rho^{(1)})_{jj}]^{\downarrow}
    = \sum_{i=1}^k\,\sum_{j=(i-1)\,2^{n-m}+1}^{i\,2^{n-m}}\,\langle v_j|\rho^{(1)}|v_j\rangle = \sum_{i=1}^k\,\Tr{\Lambda_i\rho^{(1)}} = \sum_{i=1}^k\,p^{(1)}_i,
\end{equation}
where the second inequality is due to Theorem~\ref{th:schur}, since $\{(\rho^{(1)})_{jj}\}_j$ are the eigenvalues of $\rho^{(1)}$. 

From~(\ref{app:eq:maj_ch1}) and~(\ref{app:eq:maj_ch2}) it follows that the distribution $\{p'^{(1)}_i\}_i$ is majorized by $\{p^{(1)}_i\}_i$ in accordance with condition~(\ref{app:eq:maj}). 
The same holds for $p'^{(2)}_i = \Tr\Lambda'_i\rho^{(2)}$ and  $p^{(2)}_i = \Tr\Lambda_i\rho^{(2)}$ because $\rho^{(2)}$ is proportional to the identity. 
By Lemma~\ref{th:div_ineq}, among all rank $2^{n-m}$ projectors, the divergence~(\ref{eq:fdiv}) assumes its maximal value on the projectors $\{\Lambda_i\}_{i=1}^{2^m}$ defined in~(\ref{eq:app:true_proj}).\end{proof}

Because of the connection between the global and the local tasks described above, the eigenvalues of the optimal observable are calculated via~(\ref{eq:remsol}) with $\alpha$ set to $1/2$, and we denote them here as $\lambda_i$.
The calculation yields
\begin{align}
            \label{eq:app:lambdas_opt_m}
            \begin{split}
                \lambda_{1 \leqslant i\leqslant 2^{n-m}}^{(m)} &= \frac12 + \frac{2}{I_c(\boldsymbol{\Lambda}, \rho_{1/2})} \frac{1 - 2^{-m}}{1 + 2^{-m}} = 1, \\
                \lambda_{i > 2^{n-m}}^{(m)} &= \frac12 - \frac{2}{I_c(\boldsymbol{\Lambda}, \rho_{1/2})} = \frac{1}{1-2^m},
            \end{split}
        \end{align}
where $I_c(\boldsymbol{\Lambda}, \rho_{1/2})$  is the classical Fisher information~(\ref{eq:app:cFI_ref}) calculated on the projectors $\boldsymbol{\Lambda}=\{\Lambda_i\}_{i=1}^{2^m}$ of~(\ref{eq:app:true_proj}).
The variance of the optimal observable
\begin{equation}
    M^*_m = \sum_{i=1}^{2^m}\,\lambda_i^{(m)}\,\Lambda_i
\end{equation}
at each point $\alpha$ takes form
\begin{equation}
    \label{eq:app:obs_opt-m}
    \Delta^2_{\rho_\alpha} M^*_m = (1-\alpha)\left( \frac{1}{2^m - 1} + \alpha \right).
\end{equation}
We stress that there is  no dependence on the parameter $r$  of the model~(\ref{eq:app:ex_dep}).

\section{Variational ans\"atze}
\label{app:sec:ansatzes}

    In this section, we describe and depict the ans\"atze used in this work.

    \subsection{Hardware-efficient ansatz}
    \label{app:hea}
    
        The first ansatz we use in our work is the hardware-efficient ansatz (HEA) \cite{kandala_hardware-efficient_2017}.
        This ansatz alternates between single-qubit rotations which are commonly considered to be easily implementable on contemporary hardware, and multi-qubit operations capable of entangling the qubits.
        The more layers $l$ this ansatz has, the more expressive it is. 
        
        In Fig.~\ref{fig:hea}, shown is an instance of HEA having $l=2$ layers with the entangling operation being a ladder of $ZZ$-rotations.
        Alongside with other cases, this ansatz is used for numerical experiments described in Section~\ref{sec:lin_comb-numerics}.
        As we study the performance of the parameter prediction with the observable \eqref{eq:obs-par} with different numbers of measured qubits, the value $m\in\{1,3,5\}$ is also indicated in the figure. 

        This ansatz is also used for numerical experiments with the Ising Hamiltonian described in Section~\ref{sec:ising-numerics} and Appendix~\ref{app:sec:ising-numerics}, as well as the Schwinger Hamiltonian studied in Appendix~\ref{app:schwinger-numerics}.

        \begin{figure}[h]
            \centering
            \includegraphics[width=0.9\linewidth]{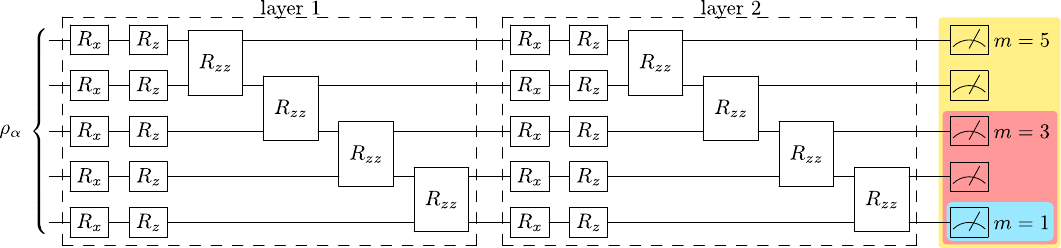}
            \caption{An instance of two-layered HEA for $n=5$ qubits and $m\in\{1,3,5\}$ measured qubits.
            The operators in the blocks are Pauli rotations $R_\sigma(\theta_j) = e^{-i\theta_j\sigma}$ with $\sigma\in\{X,Z,ZZ\}$ being a Pauli string and the rotation angles $\theta_j$ omitted.
            }
            \label{fig:hea}
        \end{figure}

    \subsection{Quantum convolutional neural networks}
    \label{app:qcnn}

        Another ansatz we consider in our work is the quantum convolutional neural network (QCNN) \cite{cong2019quantum}.
        In this ansatz, one alternates between convolutional layers and pooling layers.  
        The former connects the neighboring qubits with two-qubit blocks, and the latter reduces the system size (commonly, by half) by tracing out qubits.
        Within each layer, the parameters in the convolutional and pooling operators are usually kept the same across the blocks.
        A QCNN used in this work is shown in Fig.~\ref{fig:qcnn-ising-schwinger}.
        This ansatz is applied for the Ising Hamiltonian in Appendix~\ref{app:sec:ising-numerics}.
        It is also used for the Schwinger Hamiltonian in Appendix~\ref{app:schwinger-numerics}, but with removing the convolutional blocks between the first and the last qubits, as it was done in \cite{nagano2023quantum}.
        In both cases, we study the performance of QCNN with $m \in \{1, 2\}$ measured qubits, as indicated in the figyre.

        \begin{figure}[h]
            \centering
            \includegraphics[width=0.425\linewidth]{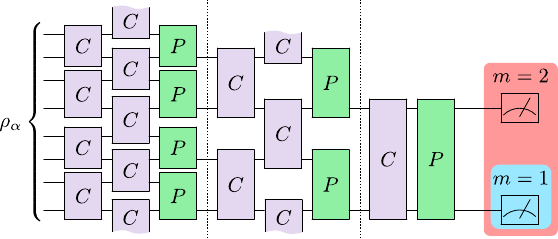}
            \hspace{0.75cm}
            \includegraphics[width=0.495\linewidth]{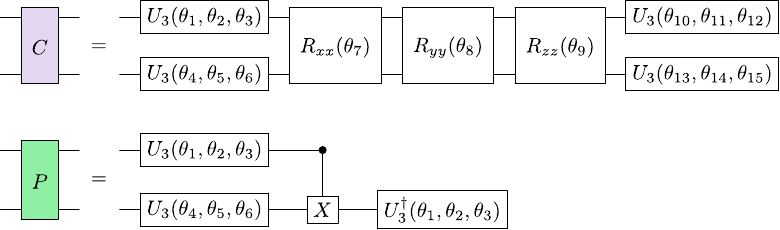}
            \caption{
                Left: Quantum convolutional neural network (QCNN) used in this work, with convolutional and pooling blocks denoted as $C$ and $P$, respectivelly; note that there are convolutional blocks connecting the first and the last qubits within each layer.
                Right: Representation of the blocks in terms of quantum gates adapted from \cite{nagano2023quantum}; here, $U_3$ are universal single-qubit rotations, and $R_\sigma(\theta_j) = e^{-i\theta_j\sigma}$ is a two-qubit rotation with  $\sigma\in\{XX,YY,ZZ\}$.
                QCNN of this form is used for the Ising Hamiltonian in Appendix~\ref{app:sec:ising-numerics}, and Schwinger Hamiltonian in Appendix~\ref{app:schwinger-numerics}; for the latter, the convolutional blocks $C$ between the first and the last qubits are removed (except for the last one, before the measurement). 
            }
            \label{fig:qcnn-ising-schwinger}
        \end{figure}

        In our work, we also use the QCNN depicted in Fig.~\ref{fig:qcnn-cluster}.
        This QCNN is taken from \cite{umeano2023can} where it was designed for classification of the ground states of the Hamiltonian we consider in Appendix~\ref{app:cluster-numerics}.
        We study the performance of this ansatz with $m \in\{1,3\}$ measured qubits.

        \begin{figure}[h]
            \centering
            \includegraphics[width=0.985\linewidth]{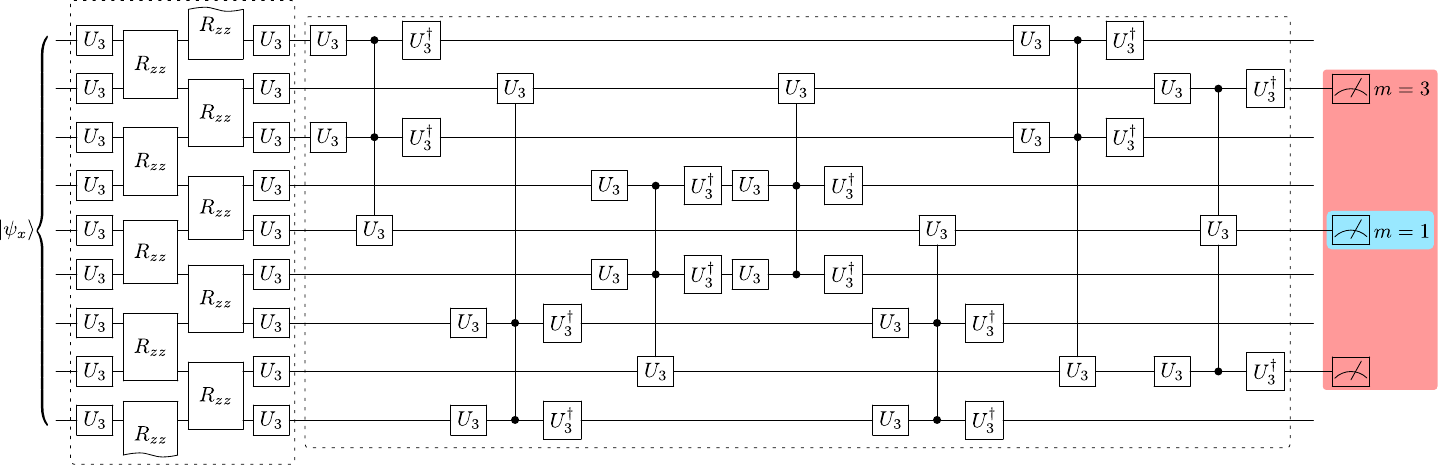}
            \caption{
            A QCNN taken from \cite{nagano2023quantum} and used in Section~\ref{app:cluster-numerics}.
            The gates are the  same as in Figure~\ref{fig:qcnn-ising-schwinger}.
            }
            \label{fig:qcnn-cluster}
        \end{figure}

    \subsection{Hamiltonian variational ansatz}
    \label{app:hva}

    Finally, in this work we also make use of the Hamiltonian variational ansatz (HVA) \cite{wiersema2020exploring}.
    Given a problem Hamiltonian as $H=\sum_q h_q H_q$ with $h_i$ real and $H_q$ Hermitian, a HVA consists of gates of the form $e^{-i \theta H_q}$.
    We apply this ansatz to the cluster Hamiltonian we study in Section~\ref{app:cluster-numerics}.
    An $l$-layered HVA for this Hamiltonian therefore becomes
    \begin{equation}
        \label{app:eq:hva}
        U_{\boldsymbol{\theta}} = \prod_{k=1}^{l} \left[ \exp\left(-i\theta_{k,3} \sum_{j=1}^n Z_j X_{j+1} Z_{j+2}\right) \exp\left(-i\theta_{k,2} \sum_{j=1}^n Z_j\right) \exp\left(-i\theta_{k,1} \sum_{j=1}^n X_j\right) \right]
    \end{equation}

\section{Additional numerical results}
\label{app:additional_numerics}

    This Section contains additional numerical results of solving regression problems for pure states.
    In these problems, we want to learn to predict the label $\alpha$ of a labeled state $\ket{\psi_\alpha}$ being the ground state of a Hamiltonian $H_\alpha$ parametrized by $\alpha$.
                
    \subsection{Transverse field Ising Hamiltonian}
    \label{app:sec:ising-numerics}

        In this Section, as in the main text, we consider the transverse field Ising Hamiltonian:
        \begin{equation}
            \label{app:eq:ising_ham}
            H_h = \sum\limits_{i=1}^n \left( Z_i Z_{i+1} + h X_i \right),
        \end{equation}
        where now we set the number of qubits $n=8$.
        Let $\ket{\psi_h}$ be the ground state of $H_h$.
        We want to learn to predict the transverse field $h$ given a state $\ket{\psi_h}$.
        Our task therefore is to train an observable $M$ giving the expectation $\langle M\rangle_{\psi_h} = h + b_h$ with small prediction bias $b_h$ and presumably low variance $\Delta^2_{\psi_h} M$.
        For this, as in the main text, we generate a training set $\mathcal{T}=\big\{\big(|\psi_{h_i}\rangle, h_i\big)\big\}_{i=1}^{10}$ and numerically solve \eqref{eq:ls_min}.

        We test two ans\"atze $U_{\boldsymbol{\theta}}$ for our task.
        First one is the QCNN described in Appendix~\ref{app:qcnn}, for which we consider the case of $m \in \{1, 2\}$ measured qubits in \eqref{eq:obs-par}.
        The second ansatz is the HEA we used in the main text, for which we measure $m=2$ qubits.

        The results of our numerical experiments are shown in Fig.~\ref{fig:ising}.
        By increasing the number of measured qubits $m$, we decrease both the prediction error and the variance. 
        We also observe that when measuring $m=2$ qubits, HEA performs better than QCNN, which may be due to the higher expressivity of the former.
        However, in some cases this ansatz is known to be prone to the phenomenon of vanishing gradients known as barren plateau \cite{leone2024practical}.

        \begin{figure*}
            \centering
            \includegraphics[width=0.475\linewidth]{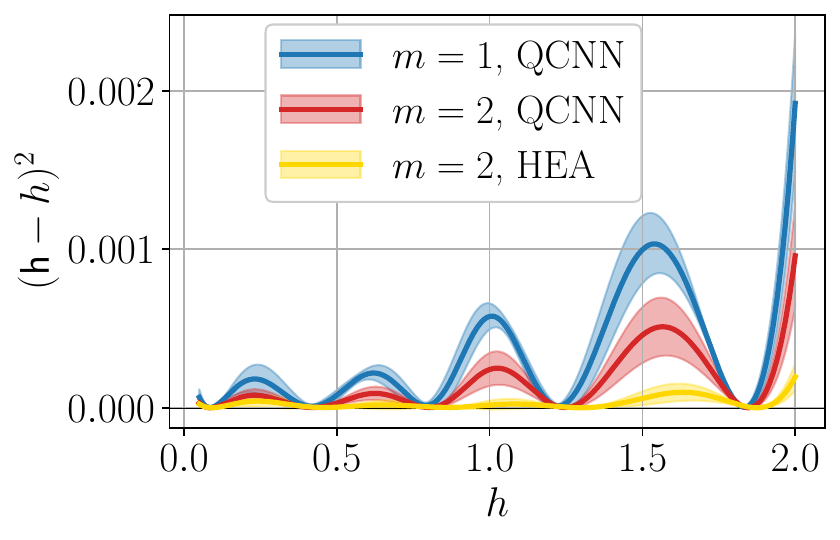}
            \includegraphics[width=0.475\linewidth]{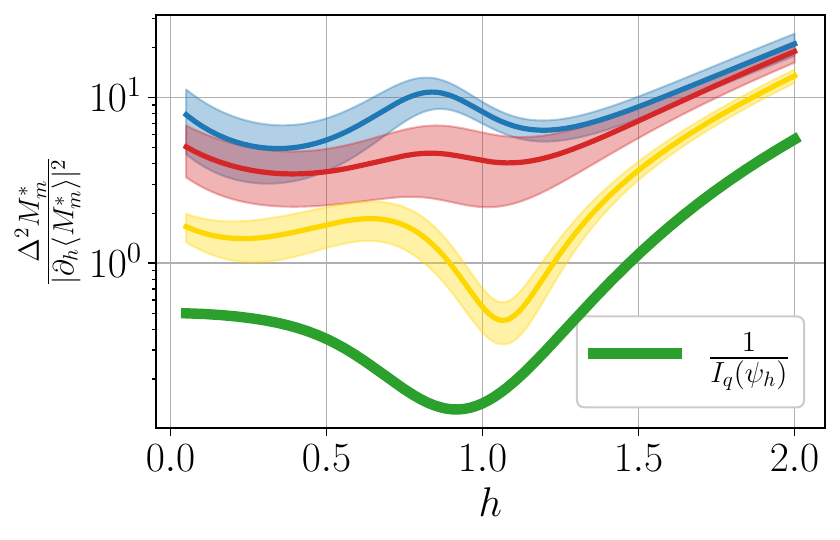}
            \caption{
                Squared difference between the prediction $\mathsf{h}=\langle M^*_m\rangle_{\psi_h}$ and the true transverse field $h$
                (left), and the variance of the optimized observable (right) obtained via numerically solving \eqref{eq:ls_min}.
                The training set is $\mathcal{T}=\big\{\big(|\psi_{h_i}\rangle, h_i\big)\big\}_{i=1}^{10}$ with $\ket{\psi_{h_i}}$ being the ground state of \eqref{app:eq:ising_ham} of $n=8$ qubits, and $h$ are picked equidistantly from $[0.05, 2]$.
                Different colors indicate different numbers of measured qubits $m$ in \eqref{eq:obs-par}, as well as the ansatz applied.
                The solid green line stands for the right-hand of the bound \eqref{eq:var-qcrb}.
                Shaded areas show the standard deviation.
            }
            \label{fig:ising}
        \end{figure*}

    \subsection{Schwinger Hamiltonian}
        \label{app:schwinger-numerics}

        Now we consider a task similar to the one solved in \cite{nagano2023quantum}, where the authors studied the Schwinger Hamiltonian \cite{kokail2019self}:
        \begin{equation}
            \label{eq:schwinger_ham}
            H_{\mathrm{Schw}}(\mu) = w \sum_{j=1}^{n-1} \left( X_j X_{j+1} + Y_j Y_{j+1} \right) + \frac{\mu}{2}\sum_{j=1}^n (-1)^j Z_j + g\sum_{j=1}^n \left( \epsilon_0 - \frac12 \sum_{l=1}^j \Big( Z_l + (-1)^j \Id \Big) \right).
        \end{equation}
        This Hamiltonian is a quantum electrodynamics model describing interacting fermions in electric field.
        The first term describes the creation/annihilation of an electron/positron pair with coupling $w$, the second is the mass term with the bare mass parameter $\mu$, and the third term is the electric field energy with coupling $g$ and background electric field $\epsilon_0$.
        Setting $w=g=1$ and $\epsilon_0=0$, the Hamiltonian $H_{\mathrm{Schw}}(\mu)$ is known to have a critical point at $\mu = \mu_c \approx -0.7$ \cite{kokail2019self}.
        In \cite{nagano2023quantum}, the authors applied a QCNN described in Appendix~\ref{app:qcnn} for telling wether the ground state $\ket{\psi_\mu}$ of $H_{\mathrm{Schw}}(\mu)$ has $\mu < \mu_c$ or $\mu > \mu_c$, i.e., solving a classification problem.

        In our work, we apply the same QCNN as in \cite{nagano2023quantum} (see Fig.~\ref{fig:qcnn-ising-schwinger}) for solving a regression task as the one considered in the main text: 
        We want to find an observable $M$ which predicts the bare mass $\mu$ with a small bias $b_\mu$, i.e., $\langle M \rangle_{\psi_\mu} = \mu + b_\mu$.
        The observable $M$ is again found by numerically solving \eqref{eq:ls_min}.
        When using a QCNN as the ansatz $U_{\boldsymbol{\theta}}$ in \eqref{eq:obs-par}, we allow to measure $m \in \{1, 2\}$ qubits.
        In addition, we consider the application of HEA described in Appendix~\ref{app:hea} with $l=7$ layers and measuring $m \in \{2, 4\}$ qubits.

        The results for both cases are shown in Fig.~\ref{fig:schwinger}. 
        As we see, by measuring an additional qubit in QCNN, we lower both the prediction error and the variance.
        When measuring $m=2$ qubits, HEA performs better.
        With $m=4$, the results obtained with this ansatz improve even more.
        However, with each measured qubit we double the number of parameters to vary, see \eqref{eq:obs-par}.

        \begin{figure*}
            \centering
            \includegraphics[width=0.475\linewidth]{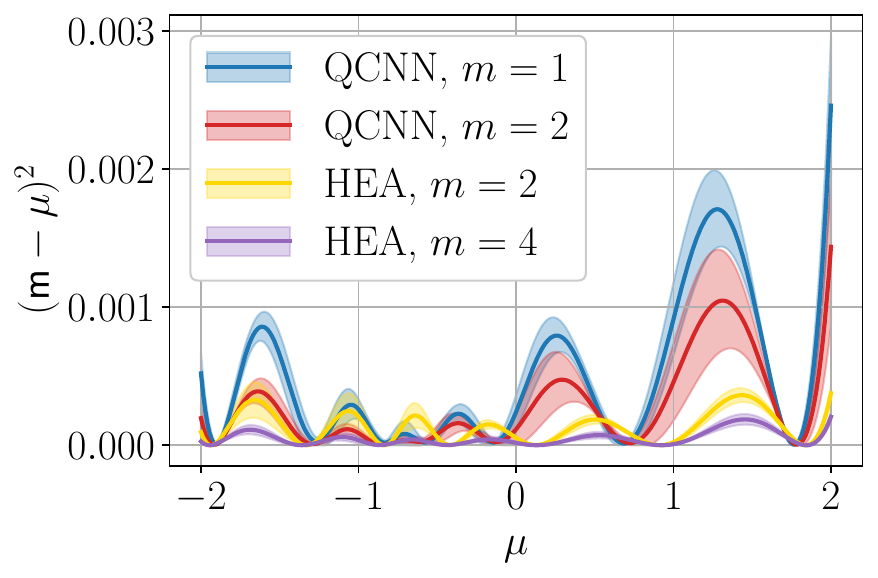}
            \includegraphics[width=0.475\linewidth]{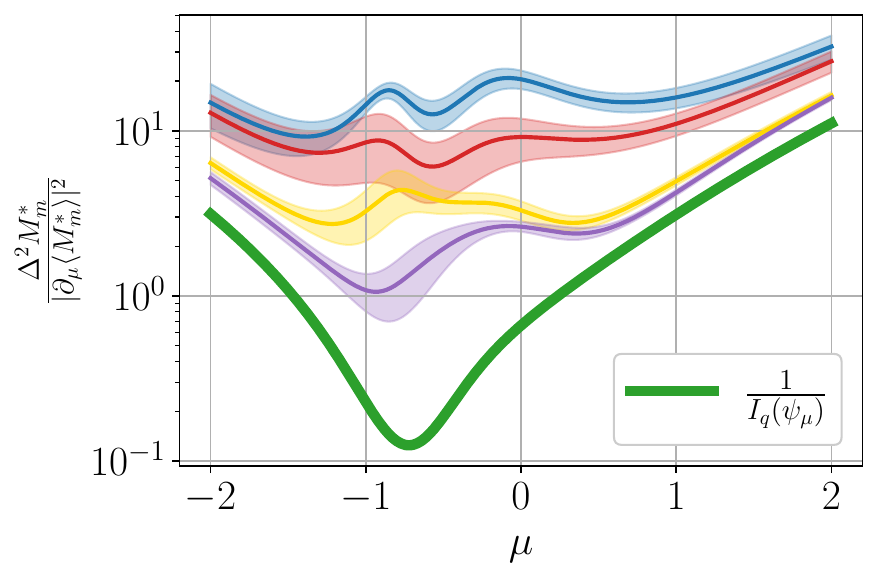}
            \caption{
                Squared difference between the predicted $\mathsf{m}=\langle M_m^*\rangle_{\psi_\mu}$ and the true bare mass $\mu$
                (left) and the variance of the optimized observable (right) obtained via numerically solving \eqref{eq:ls_min}.
                The training set is $\mathcal{T}=\big\{\big(|\psi_{\mu_i}\rangle, \mu_i\big)\big\}_{i=1}^{10}$ with $\ket{\psi_{\mu_i}}$ being the ground state of \eqref{eq:schwinger_ham} of $n=8$ qubits, and $\mu$ are picked equidistantly in $[-2, 1]$.
                Different colors indicate different numbers of measured qubits $m$ in \eqref{eq:obs-par}, as well as the ansatz applied.
                The solid green line stands for the right-hand side of the bound \eqref{eq:var-qcrb}.
                Shaded areas show the standard deviation.
            }
            \label{fig:schwinger}
        \end{figure*}

    \subsection{Reparametrized cluster Hamiltonian}
    \label{app:cluster-numerics}

        Finally, consider the Hamiltonian of the following form \cite{umeano2023can}:
        \begin{equation*}
            H_{\mathrm{cluster}} = -J \sum_{i=1}^n Z_i X_{i+1} Z_{i+2} - h_x \sum_{i=1}^n X_i - h_z \sum_{i=1}^n Z_i.
        \end{equation*}
        The first term here is the cluster Hamiltonian, the second describes the transverse field, and the third is introduced to remove the degeneracy of the ground state at $h_x = 0$. 
        This Hamiltonian can be reparametrized as
        \begin{equation}
            \label{eq:cluster-ham}
            H_{\mathrm{cluster}}(x) = -\cos\left(\frac{\pi x}{2} \right) \sum_{i=1}^n Z_i X_{i+1} Z_{i+2} -\sin\left(\frac{\pi x}{2} \right) \sum_{i=1}^n X_i - \varepsilon\sum_{i=1}^n Z_i,
        \end{equation}
        where 
        \begin{equation*}
            x = \frac{2}{\pi} \arcsin\left( \frac{h_x}{\sqrt{J^2 +h_x^2}} \right), \qquad \epsilon = \frac{h_z}{\sqrt{J^2 +h_x^2}}.
        \end{equation*}
        Considering $x \in [0, 1]$ and keeping $\varepsilon$ small, in \cite{umeano2023can}, the authors applied a QCNN for classifying the ground states $\ket{\psi_x}$ of $H_{\mathrm{cluster}}(x)$ into two classes: when $x <0.5$ and $x>0.5$.

        Setting $\varepsilon = 10^{-2}$, we solve a regression problem for predicting $x$ given a ground state $\ket{\psi_x}$.
        We used the QCNN shown in Fig.~\ref{fig:qcnn-cluster} and allowed measuring $m=1$ and $m=3$ qubits.
        As we see in Fig.~\ref{fig:cluster}, measuring $m=3$ qubits is again gives smaller prediction error and lower bias.
        For comparison, we also show the results obtained using the Hamiltonian variational ansatz (HVA) defined in \eqref{app:eq:hva} with $l=10$ layers and $m=3$ measured qubits.
        In our task, HVA also showed better results compared to QCNN.

        \begin{figure*}
            \centering
            \includegraphics[width=0.475\linewidth]{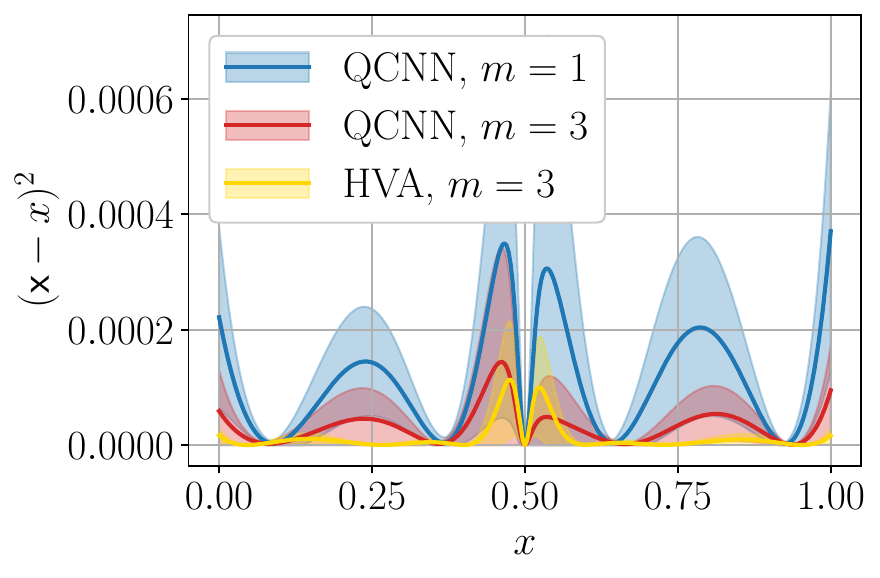}
            \includegraphics[width=0.475\linewidth]{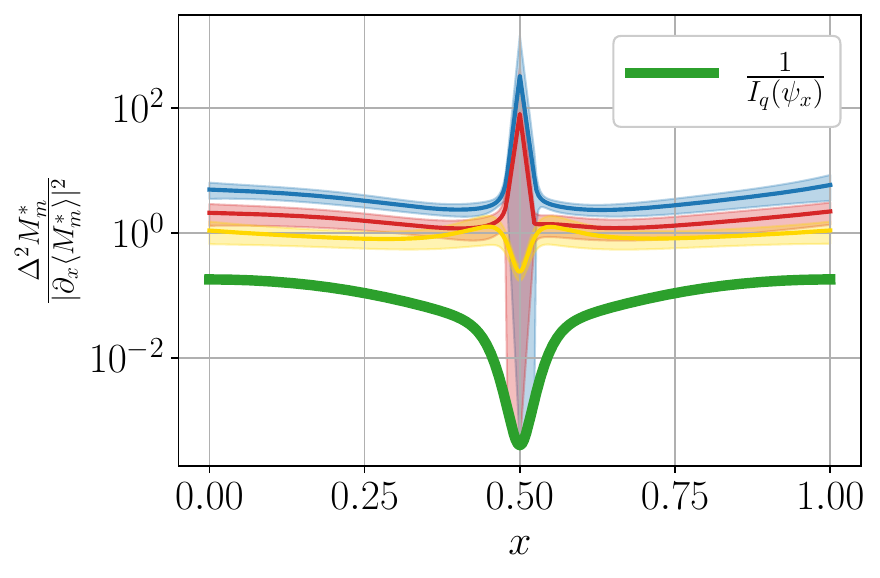}
            \caption{
                Squared difference between the predicted $\mathsf{x}=\langle M_m^*\rangle_{\psi_x}$ and the true parameter $x$ of the cluster Hamiltonian
                (left) and the variance of the optimized observable (right) obtained via numerically solving \eqref{eq:ls_min}.
                The training set is $\mathcal{T}=\big\{\big(|\psi_{x_i}\rangle, x_i\big)\big\}_{i=1}^{10}$ with $\ket{\psi_{x_i}}$ being the ground state of \eqref{eq:schwinger_ham} of $n=8$ qubits, and $x$ are picked equidistantly in $[0, 1]$.
                Different colors indicate different numbers of measured qubits $m$ in \eqref{eq:obs-par}, as well as the ansatz applied.
                The solid green line stands for the right-hand side of the bound \eqref{eq:var-qcrb}.
                Shaded areas show the standard deviation.
            }
            \label{fig:cluster}
        \end{figure*}

\bibliography{bibliography}
\bibliographystyle{unsrt}

\end{document}